\documentclass[twocolumn]{autart}    

\usepackage{amsmath,amssymb,amsfonts,stmaryrd,mathtools}
\usepackage{algorithm,algpseudocodex}
\usepackage{graphicx,placeins}

\usepackage{multirow,booktabs}

\usepackage{xparse}
\usepackage{marginnote}

\NewDocumentCommand{\hlight}{m O{0pt} m}{%
  \color{blue}{%
    \marginnote{%
      \makebox[25pt]{\footnotesize\textbf{#1}}%
    }[#2]%
    #3%
  }%
  \color{black}
}

\usepackage[bookmarks=false]{hyperref}

\usepackage{tikz}
\usetikzlibrary {shapes,shapes.geometric,arrows.meta,hobby}

\usepackage{enumitem}

\usepackage{subfig}
\usepackage{textcomp}
\usepackage{colortbl}
\definecolor{darkred}{RGB}{139,0,0}

\newcommand{\be}{\begin{equation}}
\newcommand{\ee}{\end{equation}}

\newcommand{\ben}{\begin{equation*}}
\newcommand{\een}{\end{equation*}}

\newcommand{\bbm}{\begin{bmatrix}}
\newcommand{\ebm}{\end{bmatrix}}

\DeclareMathOperator{\interior}{int}
\DeclareMathOperator{\li}{Li}

\DeclareMathOperator{\col}{col}

\newcommand{\bBm}{\begin{Bmatrix}}
\newcommand{\eBm}{\end{Bmatrix}}

\newcommand{\bvm}{\begin{vmatrix}}
\newcommand{\evm}{\end{vmatrix}}

\newcommand{\bVm}{\begin{Vmatrix}}
\newcommand{\eVm}{\end{Vmatrix}}

\newcommand{\bpm}{\begin{pmatrix}}
\newcommand{\epm}{\end{pmatrix}}

\newcommand{\bnm}{\begin{matrix}}
\newcommand{\enm}{\end{matrix}}

\newcommand{\bi}{\begin{itemize}}
\newcommand{\ei}{\end{itemize}}

\newcommand{\eor}{\ensuremath{\hfill\blacksquare}}

\newtheorem{remark}{Remark}
\newtheorem{lemma}{Lemma}
\newtheorem{proposition}{Proposition}
\newtheorem{corollary}{Corollary}

\usepackage{wrapfig}

\renewenvironment{IEEEbiography}[2][]
{%
  \par\medskip
  \footnotesize
  \setlength{\intextsep}{0pt}%
  \setlength{\columnsep}{4pt}%
  \begin{wrapfigure}[9]{l}{1.1in}
    \vspace{-0.2\baselineskip}%
    \centering
    \ifx\relax#1\relax
      \framebox(1in,1.25in){}%
    \else
      #1%
    \fi
    \vspace{0pt}%
  \end{wrapfigure}%
  \noindent\textbf{#2}\enspace
}
{%
  \par\vspace{10pt}%
}

\begin{document}

\begin{frontmatter}

\title{\bf Safe Navigation in Cluttered Environments \\ Via Spline-Based Harmonic Potential Fields \thanksref{footnoteinfo}} 

\thanks[footnoteinfo]{This paper was not presented at any IFAC 
meeting. Corresponding author: F. Stoican, Tel. +40760148950. 
}

\author[1]{Theodor-Gabriel Nicu}\ead{theodor.nicu@upb.ro},    
\author[1]{Florin Stoican}\ead{florin.stoican@upb.ro},               
\author[1]{Daniel Ioan}\ead{daniel.ioan@upb.ro},  
\author[2]{Ionela Prodan}\ead{ionela.prodan@lcis.grenoble-inp.fr}

\address[1]{National University of Science and Technology Politehnica Bucharest, CAMPUS Research Institute, Department of Automatic Control and Systems Engineering, Romania, 313 Splaiul Independen\c{t}ei, 060042, Sector 6, Bucharest, Romania}  

\address[2]{Univ. Grenoble Alpes, Grenoble INP$^\dagger$, LCIS, F-26000, Valence, France. $^\dagger$ Institute of Engineering and Management, Univ. Grenoble Alpes.}             

\begin{abstract}     


We provide a motion-planning framework for target reaching and obstacle avoidance in cluttered environments. Given a polyhedral decomposition of the feasible space, the agent is guided through a prescribed sequence of cells by a suitable control policy. For each cell, we construct a potential induced by a harmonic function on the unit disk whose Dirichlet boundary condition is represented by a cardinal B-spline curve. A detailed analysis of the curve periodicity, support, and control-point selection provides an explicit analytical construction of these potentials and the associated vector fields. We further characterize inter-cell control continuity and forward invariance of non-exit facets, and introduce a smoothed control law ensuring continuous switching between independently constructed neighboring cells. Numerical results illustrate that the resulting framework safely guides the agent along the prescribed cell sequence from the initial point to the target.

\end{abstract}

\end{frontmatter}
\vspace{-1em}
\section{Introduction}
\label{sec:intro}
\vspace{-1em}
Motion planning \cite{lavalle2006planning,dahlin2025autonomous,zou2025cooperative} remains one of the main topics of interest in control theory and autonomous systems. In this context, the \emph{Artificial Potential Field} (APF) approach \cite{khatib1986real,rimon1991construction,paternain2017navigation} is both popular and versatile. Broadly speaking, the idea is to construct a potential surface that encodes attraction to the target and repulsion from obstacles, and to apply a control action that drives the system along the negative gradient of this surface. In complex environments, however, trajectories may become trapped in local minima, issue which has spanned multiple approaches. Ad-hoc modifications of the surface, such as adding virtual hills \cite{park2004real} or reactive components \cite{haddad1998reactive}, often work in practice. Navigation functions are attractive because they guarantee the existence of a unique, nondegenerate global minimum
\cite{khatib1986real,filippidis2011adjustable,rimon1991construction,wang2025hybrid},
but this comes at the price of convoluted formulations and impractical parameter bounds, e.g., the tuning parameter $k$ used in
\cite{rimon1991construction,filippidis2011adjustable}.
Another direction, which avoids some of the difficulties of global (one-step) methods, is to decompose the motion-planning problem into two
stages. First, the feasible space is partitioned into a collection of disjoint regions, typically polyhedral \cite{mahulea2020path}; second,
local navigation laws drive the agent successively between neighboring regions to avoid obstacles and, ultimately,
reach the target. This strategy is adopted in
\cite{conner2003construction}, where the local navigation mechanism is
based on harmonic functions, i.e., solutions of the Laplace equation
\cite{axler2013harmonic}, which also arises, for example, in the
description of steady-state heat propagation
\cite{ma2025closed}. Since a nonconstant harmonic function has no interior extrema, its minima and maxima are attained on the boundary of the region. Appropriately selecting nonconstant boundary
values therefore induces the desired flow, which can be directed toward a prescribed exit facet of the polyhedral region. The corresponding
harmonic solution admits a closed-form representation in polar coordinates on the unit disk \cite{conner2009flow}; through a diffeomorphic transformation \cite{tanner2003nonholonomic}, it can then
be transferred to arbitrary polyhedral cells, as in
\cite{conner2003construction}. The same construction may
be extended to higher dimensions using
hyperspherical harmonics and corresponding boundary
parametrizations \cite{masoud2012harmonic}, which is beyond the scope of this 2D
study.

The main limitation of this harmonic-navigation construction, and the gap addressed here, is the restrictive piecewise specification of the boundary conditions, which may induce sharp changes in the surface gradient near the polyhedron vertices and offers limited flexibility in shaping the flow. In our preliminary work \cite{nicu2024safe}, we introduced a framework in which the boundary curves are represented as weighted sums of cardinal B-splines \cite{lyche2018foundations}. The resulting harmonic surfaces were shown to funnel the agent from the initial point to the target along a sequence of cells. The present work extends \cite{nicu2024safe} in four directions: i) periodicity and support conditions for the boundary function are derived, clarifying its tuning parameters; ii) a systematic procedure for selecting the B-spline control points is provided; iii) the harmonic potential formulation is generalized to arbitrary B-spline order; iv) the control policy is further analyzed in terms of inter-cell vector-field matching, forward invariance, and smoothing of cell transitions.


\noindent\emph{Notation}: Let $\mathbb X_{>n}=\{x\in\mathbb X:x>n\}$, with $\mathbb X\in\{\mathbb Z,\mathbb R\}$. For a compact set $\mathcal C$, $\interior\mathcal C$ and $\partial\mathcal C$ denote its interior and boundary, respectively. $\mathfrak F(\cdot)$ denotes the Fourier transform, while $\partial f/\partial x$ (or $df/dx$ for univariate $f$) denotes differentiation with respect to $x$. Moreover, $\|x\|$ is the Euclidean norm, $\lceil x\rceil$ and $\lfloor x\rfloor$ are the ceiling and floor functions, and $\operatorname{sgn}(x)$ denotes the sign of $x$. $\nabla_x\phi$ denotes the gradient of $\phi$ with respect to $x$, with the usual chain rule applied to composed functions.

\section{Problem statement}
\label{sec:statement}
\vspace{-1em}
We address the motion planning problem in two stages: i) decompose the free space into a union of cells and ii) prescribe, via local harmonic functions, the behavior within each cell so as to ensure a feasible path from the start point to the destination.  All illustration, consider a space cluttered with obstacles $\mathcal O_j$, with $\mathcal W\setminus\bigcup_j\mathcal O_j
=\bigcup_\ell\mathcal P_\ell$, 
where $\mathcal W$ is a compact polytopic workspace, the obstacles $\mathcal O_j$ are polytopes\footnote{\label{foot:decomposition} Otherwise, finitely many cells may not suffice (e.g., $\mathbb{R}^2\setminus\mathcal{E}$, $\mathcal{E}$ an ellipsoid); smooth obstacles may be over-approximated by polytopes.}, and the cells $\mathcal P_\ell$ are polytopes with nonempty, pairwise disjoint interiors forming a polyhedral complex, obtained, e.g., as in \cite{mahulea2020path}. Hence, obstacle avoidance reduces to maintaining the trajectory within a prescribed sequence of feasible cells and crossing only their selected exit facets.
\begin{figure}[!ht]
    \centering
    \includegraphics[width=.8\columnwidth]{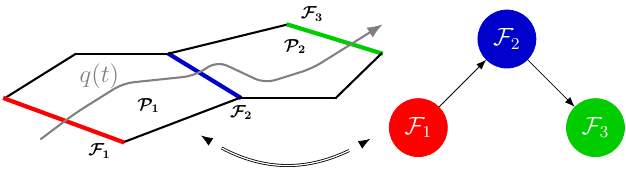}
    \caption{Proof of concept illustration}
    \label{fig:proof_concept}
\end{figure}
To illustrate this idea, consider two cells, $\mathcal{P}_{1,2}$ and three facets, $\mathcal{F}_{1,2,3}$ chosen such that $\mathcal{F}_1\in \mathcal{P}_1\setminus \mathcal{P}_2$, $\mathcal{F}_2=\mathcal{P}_1\cap \mathcal{P}_2$ and $\mathcal{F}_3=\mathcal{P}_2\setminus \mathcal{P}_1$. Suppose the agent enters $\mathcal{P}_1$ through facet $\mathcal{F}_1$, transitions into $\mathcal{P}_2$ through $\mathcal{F}_2$, and exits $\mathcal{P}_2$ through $\mathcal{F}_3$, as shown in Fig.~\ref{fig:proof_concept}. 
~The remainder of the paper develops a modified harmonic surface
construction with sufficient conditions for safe transitions
between cells.

\vspace{-1em}
\section{Preliminaries on cardinal B-splines}
\label{sec:prelim}
\vspace{-1em}

To characterize the potential field within each cell, we consider cardinal B-splines whose properties are detailed next, \cite{lyche2018foundations}. The cardinal splines $B_{k,p,\mathbb Z}(t)$, of degree $p\geq 1$, are obtained recursively from relations
\begin{subequations}
\label{eq:cardinal_splines}
\begin{align}
\mkern-12mu\label{eq:cardinal_splines_a} M_0(\theta)&=\begin{cases}1, & \forall \theta \in [0,1),\\ 0, &\text{otherwise},\end{cases}\\
\mkern-12mu\label{eq:cardinal_splines_b} M_p(\theta)&=\frac{\theta}{p}M_{p-1}(\theta)+\frac{p+1-\theta}{p}M_{p-1}(\theta-1),\\
\mkern-12mu\label{eq:cardinal_splines_c} B_{k,p,\mathbb Z}(\theta)&=B_{0,p,\mathbb Z}(\theta-k)=M_p(\theta-k), \quad \forall k \in \mathbb Z,
\end{align}
\end{subequations}
\noindent with the properties of (given in `local support' form \cite{lyche2018foundations}):
\begin{enumerate}[label=\textbf{P\arabic*)},nosep]
    \item\label{prop:support} Local support: $B_{k,p,\mathbb Z}(\theta)=0,\quad \forall \theta\notin [k,k+p+1)$;
    
    \item\label{prop:active_interval} Partition of unity:\\ $\sum_{k=m-p}^m B_{k,p,\mathbb Z}(\theta)=1$, $\forall \theta\in [m, m+1)$;
    \item\label{prop:fourier} The Fourier transform of \eqref{eq:cardinal_splines_b} is given by\\ $\mathfrak F\bigl\{M_p(\theta)\bigr\}(\omega)=\left(\frac{1-e^{-j\omega}}{j\omega}\right)^{p+1}$.
    \item\label{prop:derivative} The derivative of \eqref{eq:cardinal_splines_b} is given by\\ $\frac{d}{d\theta}M_p(\theta)=M_{p-1}(\theta)-M_{p-1}(\theta-1)$.
\end{enumerate}
The standard cardinal B-spline family \eqref{eq:cardinal_splines_c} may be scaled with an arbitrary factor $\lambda\in \mathbb R_{>0}$, to obtain:
\begin{equation}
\label{eq:sigma}
\sigma_k(\theta)=B_{k,p,\mathbb Z}\bigl({\theta}/{\lambda}\bigr)=M_p\bigl(({\theta}/{\lambda })-k\bigr),
\end{equation}
which, weighted with control points $P_k$, gives the B-spline curve
\begin{equation}
\label{eq:dirichlet_spline}
    h(\theta)=\sum\limits_{k\in \mathbb Z}P_k\sigma_k(\theta),\quad \forall \theta\in \mathbb R.
\end{equation}
\vspace{-2em}
\subsection*{Illustrative example}
\vspace{-1em}

In Fig.~\ref{fig:card_splines} we illustrate, for $p=3$, both the seed $M_\ell(\theta)$ functions, with $0\leq \ell\leq p$ computed as in \eqref{eq:cardinal_splines_a}--\eqref{eq:cardinal_splines_b} and two of the cardinal B-splines associated with $M_p(\theta)$, obtained as in \eqref{eq:cardinal_splines_c}, for $k\in \{-4,3\}$.


\begin{figure}[!ht]
    \centering
    \includegraphics[width=.85\columnwidth]{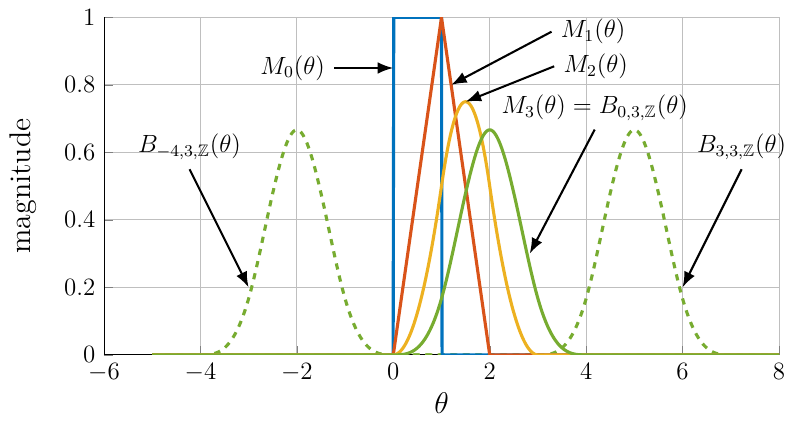}\hfill
    \caption{Cardinal B-splines}
    \label{fig:card_splines}
\end{figure}

\section{Cardinal B-spline potential surface}
\label{sec:bsplines}

The Laplace operator in $\mathbb R^2$ with polar coordinates\footnote{Variables $r,\theta$ denote the polar coordinates (radius and angle).} 
\begin{equation}
\nabla^2 = \dfrac{\partial^2}{\partial r^2} + \frac{1}{r}\dfrac{\partial}{\partial r}+\frac{1}{r^2}\frac{\partial^2}{\partial \theta^2}
\end{equation}
ensures that any function $\gamma(r,\theta)$ defined over a closed region $\mathcal C\subset \mathbb R^2$ with non-empty interior, for which 
\begin{equation}
\label{eq:laplace}
\nabla^2 \gamma(r,\theta)=0, \:\forall(r,\theta)\in\interior {\mathcal C}
\end{equation}
holds, enjoys two equivalent properties:
\begin{enumerate}[label=\roman*),itemsep=2pt]
    \item $\gamma(r,\theta)$ has no points of local minima / maxima (only saddle points), $\forall (r,\theta)\in \interior \mathcal C$; 
    \item all local minima / maxima of $\gamma(r,\theta)$ are found on the boundary of its domain:
    \begin{equation}
      \min/\max_{(r,\theta)\in \mathcal C}\gamma(r,\theta)=\min/\max_{(r,\theta)\in \partial\mathcal C}\gamma(r,\theta).
    \end{equation}
\end{enumerate}
Taking 
\begin{equation}
\label{eq:cdomain}
    \mathcal C=\{(r,\theta):\:0\leq r\leq 1,\: 0\leq \theta\leq 2\pi\},
\end{equation}
as the unit disk, and assuming that the solution over its boundary, $\partial \mathcal C=\{(r,\theta):\: r=1, \: 0\leq\theta\leq 2\pi\}$, is known (the ``Dirichlet boundary problem''), $\gamma(r=1,\theta)=h(\theta)$, the surface (i.e. the graph of $\gamma(r,\theta)$) that verifies $\nabla^2\gamma(r,\theta)=0,\:\forall (r,\theta)\in\interior \mathcal C$ is\footnote{\label{foot:poisson}Noteworthy, trigonometric representations of the boundary function as in \eqref{eq:h_fourier} allow an equivalent Poisson kernel representation of the solution, \cite[Sec. 4.2.3]{conner2003construction}.}
\begin{equation}
\label{eq:harmonic_fourier}
    \gamma(r,\theta)=A_0+\sum\limits_{n=1}^\infty A_nr^n\cos n\theta + B_nr^n\sin n\theta,
\end{equation}
where $A_n,B_n$ come from the trigonometric Fourier series, as detailed in \cite[Sec. 4.5]{powers2014boundary}:
\begin{equation}
\label{eq:h_fourier}
    h(\theta)=\gamma(1,\theta)=A_0+\sum\limits_{n=1}^\infty A_n\cos n\theta + B_n\sin n\theta.
\end{equation}

While \eqref{eq:harmonic_fourier} is computed over the unit disk, any map $\varphi(\cdot)$ that is full rank in the interior of an arbitrary convex domain, pulls \eqref{eq:harmonic_fourier} to it, preserving local minima/maxima, \cite[Sec.~5.1]{conner2003construction}. Consider a polytope given as
\begin{equation}
\label{eq:cell}
    \mathcal P=\{q\in \mathbb R^2: \beta_i(q)\geq 0, \forall i=1,\ldots,m\}, 
\end{equation}
 where $\beta_i(q)=a_i^\top q -b_i$ are linear functions chosen such that $\mathcal P$ is not empty (with shorthand notation $q=[x\: y]^\top$). Then, taking as in \cite[Sec. 5.1]{conner2003construction}: 
\begin{equation}
\label{eq:betaq}
\beta(q)=\beta_{\text{max}}^{\frac{1-m}{m}}\prod\limits_{i=1}^m \beta_i(q),
\end{equation}
with 
\begin{equation}
\label{eq:betamax}
\beta_{\text{max}}=\prod\limits_{i=1}^m \beta_i(q_\beta), \quad \text{and}\quad q_\beta=\arg\max\limits_q \prod\limits_{i=1}^m \beta_i(q)
\end{equation}
allows to define mapping
\begin{equation}
\label{eq:pulled}
    \varphi(q)=\frac{q}{\|q\|+\beta(q)}.
\end{equation}
Thus, $q\in \mathcal P\implies \varphi(q)\in \mathcal C$ and, moreover, $q\in \partial \mathcal P\implies\varphi(q)\in \partial\mathcal C$ and we may consider the pulled-back potential over the polytope $\mathcal P$, given as\footnote{\label{foot:notation}Hereinafter, with a slight abuse of notation, we use the same symbol $\gamma$ to denote both the polar and Cartesian representations.} 
\begin{equation}
\label{eq:gamma_pulled}
    \gamma_{\mathcal P}(q)=\gamma(\varphi(q)).
\end{equation}
The choice of $\varphi$ affects the resulting vector field through $\nabla_q\gamma_{\mathcal P}(q)=J_\varphi(q)^\top\nabla\gamma(\varphi(q))$, and therefore different mappings may induce different geometric distortions and trajectories. Alternative constructions include Schwarz--Christoffel conformal mappings, which locally preserve angles but require a more involved numerical construction \cite{driscoll2002schwarz,fan2023using}. Here, \eqref{eq:pulled} is retained for its explicit form and the convenient automatic evaluation of compositions and derivatives within CasADi \cite{Andersson2019}.

\subsection{Computation of the harmonic function}
\label{sec:harmonic_computation}
\vspace{-1em}
Firstly, we introduce a periodicity condition for $h(\theta)$ defined as in \eqref{eq:dirichlet_spline} and characterizing the boundary of $\mathcal C$.
\begin{lemma}
\label{lem:periodicity}
For a given $T \in \mathbb{Z}_{> 0}$, let $\lambda = 2\pi / T$ and $P_{k} = P_{k\pm T}$ for all $k \in \mathbb{Z}$. Then, $h(\theta)$ defined as in~\eqref{eq:dirichlet_spline}$\,$ satisfies
\begin{equation}
\label{eq:h_periodic}
    h(\theta) = h(\theta + 2\pi), \quad \forall\, \theta \in \mathbb{R}.
\end{equation}
\end{lemma}
\vspace{-2em}
\begin{pf}
The following chain of equalities holds:
\begin{align*}
    h(\theta + 2\pi)\overset{\eqref{eq:dirichlet_spline}}{=}&\sum\limits_{k\in \mathbb Z}P_k\sigma_k(\theta+2\pi)\overset{\eqref{eq:sigma}}{=}\sum\limits_{k\in \mathbb Z}P_kM_p\biggl(\frac{\theta+2\pi}{\lambda}-k\biggr) \\ &
    \overset{\lambda = 2\pi / T}{=}\sum\limits_{k\in \mathbb Z}P_kM_p\biggl(\frac{\theta}{\lambda}+T-k\biggr)\\
    \overset{P_k=P_{k+T}}{=}&\sum\limits_{k\in \mathbb Z}P_{k+T}M_p\biggl(\frac{\theta}{\lambda}+T-k-T\biggr)\\
    &=\sum\limits_{k\in \mathbb Z}P_{k+T}M_p\biggl(\frac{\theta}{\lambda}-k\biggr)=h(\theta),
\end{align*}
where, in the last step, a change of variable $k \mapsto k - T$ is performed.  
This confirms~\eqref{eq:h_periodic}, completing the proof.\qed
\end{pf}
The next result conditions the indices of the spline functions \eqref{eq:sigma} which do not vanish over the interval $[0,2\pi]$.
\begin{lemma}
\label{lem:support}
With the notation of Lemma~\ref{lem:periodicity}, the sequence of indices $k$ whose spline functions $\sigma_k(\theta)$, defined as in \eqref{eq:sigma}, have supports
\begin{enumerate}[label=\roman*)]
    \item fully contained in $[0,2\pi]$ is
    \begin{equation}
    \label{eq:kin}
        \mathcal K_\bullet = \{0,\ldots, T-p-1\};
    \end{equation}
    \item partially contained in $[0,2\pi]$ is
        \begin{equation}
        \label{eq:kpartial}
        \mathcal K_\circ = \{-p,\ldots, -1\}\cup \{T-p,\ldots, T-1\}.
    \end{equation}    
\end{enumerate}
\end{lemma}
\begin{pf}
By property~\ref{prop:support} we have 
\begin{equation}
\label{eq:sigma_support}
    \sigma_k(\theta)=0,\ \forall \theta \notin [\lambda k, \lambda(k+p+1)),
\end{equation}
which implies that only splines with indices $k\in\{-p,\ldots, T-1\}$ contribute to the function $h(\theta)$ on the interval $[0,2\pi]$. Among these, only the indices in \eqref{eq:kin} correspond to splines whose supports lie entirely within $[0,2\pi]$ (i.e., there is no $\theta\notin [0,2\pi]$ such that $\sigma_k(\theta)\neq 0$). Subtracting this set from the former yields \eqref{eq:kpartial}.\hfill\hfill\qed
\end{pf}

We now have the tools to provide the representation of $h(\theta)$ as a Fourier trigonometric series. 
\begin{proposition}
\label{prop:fourier_coefficients}
    For a sequence of control points respecting, as in Lemma~\ref{lem:periodicity},
    \begin{equation}
    \label{eq:control_points}
        P_k = P_{k\pm T},\:\forall k \in \mathbb Z, \quad P_{\ell}=0,\: \forall \ell \in \mathcal K_\circ,
    \end{equation}
    with $\mathcal K_\circ$ as in \eqref{eq:kpartial} and $T\in \mathbb Z_{>p+1}$,     the Fourier coefficients of the function $h(\theta)$ are given by relations
\begin{subequations}
\label{eq:fourier_coefficients}
    \begin{align}
   \mkern-7mu A_0&=\frac{\lambda}{2\pi}\sum\limits_{k\in \mathcal K_\bullet}P_k,\\
    \mkern-32mu A_n&\mkern-4mu=\frac{\lambda}{2\pi(j\lambda n)^{p+1}\mkern-5mu}\sum\limits_{k\in \mathcal K_\bullet}\mkern-5mu P_k\sum\limits_{\ell=0}^{p+1}\varrho(p+1,\ell)E_{n,\ell,k}^+,\\
    \mkern-32mu B_n&\mkern-4mu=\frac{\lambda}{2\pi(j\lambda n)^{p+1}\mkern-5mu}\sum\limits_{k\in \mathcal K_\bullet}\mkern-5mu P_k\sum\limits_{\ell=0}^{p+1}\varrho(p+1,\ell)E_{n,\ell,k}^-.
    \end{align}
\end{subequations}    
with shorthand notations
\begin{multline}
    E_{n,\ell,k}^\pm=\bigl[1\pm (-1)^{p+1}\bigr]\cos \lambda n(\ell+k)\\
    \mp j\bigl[1\mp (-1)^{p+1}\bigr]\sin \lambda n(\ell+k),
\end{multline}
and $\varrho(p,\ell)=(-1)^\ell{p\choose \ell}$.
\end{proposition}
\begin{pf}
    We start by computing $\Sigma_k(\omega)$, the Fourier transform of $\sigma_k(\theta)$, from \eqref{eq:sigma}, with the help of \ref{prop:fourier}:
\begin{equation}
\label{eq:sigma_fourier}
\begin{split}
    \Sigma_k(\omega)&=\mathfrak F\bigl\{\sigma_k(\theta)\bigr\}=\mathfrak F\biggl\{ M_p\left(\frac{\theta}{\lambda}-k\right)\biggr\}\\
    &=\left(\frac{1-e^{-j\lambda \omega}}{j\lambda\omega}\right)^{p+1}\cdot |\lambda| e^{-j\lambda\omega k}.
\end{split}
\end{equation}
Under \eqref{eq:control_points}, Lem.~\ref{lem:periodicity} ensures that $h(\theta)$ is $2\pi$-periodic and Lem.~\ref{lem:support} allows to state that 
    $g(\theta)=\sum_{k\in \mathcal K_\bullet} P_k\sigma_k(\theta)$
is the restriction of $h(\theta)$ on the interval $[0,2\pi]$, i.e., that $h(\theta)=\sum_{k\in \mathbb Z} g(\theta+2\pi k)$. Thus, via \eqref{eq:sigma_fourier}, 
\begin{equation}
\label{eq:h_fourier2}
    \mkern-9mu G(\omega)\mkern-2mu=\mkern-2mu\mathfrak F\{g(\theta)\}\mkern-4mu=\mkern-9mu\sum\limits_{k\in \mathcal K_\bullet}\mkern-6mu P_k\biggl(\frac{1-e^{-j\lambda \omega}}{j\lambda\omega}\biggr)^{p+1}\mkern-16mu\cdot |\lambda| e^{-j\lambda\omega k}.
\end{equation}

Next, we make use of two results from the folklore of Fourier analysis \cite[Sec. 4.1]{oppenheim1997signals} linking\footnote{We work under the assumption that $h(\theta)$ is $2\pi$-periodic, hence its fundamental pulsation is $\omega_0=1$.} $h(\theta)$ and $G(\omega)$:
\begin{enumerate}[label=\roman*),itemsep=2pt,nosep,leftmargin=-0pt]
    \item the coefficients of the exponential form Fourier series ($D_n$) depend on those in trigonometric form ($A_n,B_n$) by
    \begin{equation}
    \label{eq:dn_equiv}
        D_0=A_0,\quad D_{\pm n}=\frac{A_n\mp jB_n}{2},\: \forall n\geq 1. 
    \end{equation}
    \item $G(\omega)$ verifies
    \begin{equation}
    \label{eq:dn_from_ft}
        D_n=\frac{\omega_0}{2\pi}G(\omega)\biggr|_{\omega=n\omega_0}\overset{\omega_0=1}{=}\frac{1}{2\pi}G(\omega)\biggr|_{\omega=n}.
    \end{equation}
\end{enumerate}
Taking $\lambda > 0$ and substituting \eqref{eq:h_fourier2} into \eqref{eq:dn_from_ft}, yields
\begin{equation}
    D_n=\frac{1}{2\pi}\sum\limits_{k\in \mathcal K_\bullet}\left[P_k\left(\frac{1-e^{-j\lambda n}}{j\lambda n}\right)^{p+1}\cdot \lambda e^{-j\lambda n k}\right],
\end{equation}
which, using \eqref{eq:dn_equiv}, yields \eqref{eq:fourier_coefficients} after a straightforward though tedious manipulation, in which the binomial theorem is applied to expand $\bigl(1 - e^{\mp j \lambda n}\bigr)^{p+1}$.
We recall the Taylor series expansion: 
\begin{equation*}
    e^{-x}=1-x+\frac{x^2}{2!}-\frac{x^3}{3!}+\frac{x^4}{4!}-...
\end{equation*}
from which we obtain the following:
\begin{align*}
    e^{-j\lambda n}=1-j\lambda n -\frac{\lambda^2n^2}{2}+\frac{j\lambda^3n^3 
    }{6}+\frac{\lambda^4n^4}{24}+...\\
    1-e^{-j\lambda n}=1-1+j\lambda n +\frac{\lambda^2n^2}{2}-\frac{j\lambda^3n^3}{6}-\frac{\lambda^4n^4}{24}-...\\
    \frac{1-e^{-j\lambda n}}{j\lambda n}=1+\frac{\lambda^2n^2}{2j\lambda n}-\frac{j\lambda^3n^3}{6j\lambda n}-\frac{\lambda^4n^4}{24j\lambda n}-...\\ 
    = 1+\underbrace{\frac{\lambda n}{2j}-\frac{\lambda^2n^2}{6}-\frac{\lambda^3n^2}{24j}-...}_{n\rightarrow 0} = 1.
\end{align*}
Therefore, we have that: 
\begin{equation*}
    \biggl(\frac{1-e^{-j\lambda n}}{j\lambda n}\biggr)^{p+1}\biggl|_{n\rightarrow0}=1
\end{equation*}
Using this result, we obtain:
\begin{equation*}
    A_0=D_0\underset{n\rightarrow 0}{=}\frac{1}{2\pi}\sum\limits_{k\in \mathcal K_\bullet}\biggl[P_k\cdot\underbrace{\biggl(\frac{1-e^{-j\lambda n}}{j\lambda n}\biggr)^{p+1}}_{=1}\cdot\lambda \underbrace{e^{-j\lambda nk}}_{=1}\biggr], 
\end{equation*}
which means that:
\begin{equation*}
    A_0=\frac{\lambda}{2\pi}\sum\limits_{k\in \mathcal K_\bullet}P_k
\end{equation*}

Next, we make use of \eqref{eq:dn_equiv} to obtain coefficients $A_n$ and $B_n$.

\begin{equation*}
    \begin{cases}
        D_n=\frac{
        A_n-jB_n}{2} \\
        D_{-n}=\frac{A_n+jB_n}{2}
    \end{cases} \overset{+}{\implies} 
        A_n=D_n+D_{-n}
\end{equation*}
\begin{align*}
A_n=\frac{1}{2\pi}\sum\limits_{k\in \mathcal K_\bullet}\biggl[P_k\cdot\biggl(\frac{1-e^{-j\lambda n}}{j\lambda n}\biggr)^{p+1}\cdot \lambda e^{-j\lambda nk}\biggr]  \\ + \frac{1}{2\pi}\sum\limits_{k\in \mathcal K_\bullet}\biggl[P_k\cdot\biggl(\frac{-1+e^{j\lambda n}}{j\lambda n}\biggr)^{p+1}\cdot \lambda e^{j\lambda nk}\biggr]  \\ 
= \frac{1}{2\pi}\sum\limits_{k\in \mathcal K_\bullet}\biggl[P_k\frac{1}{(j\lambda n)^{p+1}}\sum\limits_{\ell=0}^{p+1}(-1)^\ell{p+1\choose \ell}e^{-j\lambda n \ell}\lambda e^{-j\lambda nk}\biggr]  \\ 
+ \frac{1}{2\pi}\sum\limits_{k\in \mathcal K_\bullet}\biggl[P_k\frac{1}{(j\lambda n)^{p+1}}\sum\limits_{\ell=0}^{p+1}(-1)^\ell{p+1\choose \ell}e^{j\lambda n \ell}\lambda e^{j\lambda nk}\biggr]  \\
=\frac{\lambda}{2\pi (j\lambda n)^{p+1}}\sum\limits_{k\in \mathcal K_\bullet}\biggl[P_k\sum\limits_{\ell=0}^{p+1}(-1)^\ell{p+1\choose \ell}e^{-j\lambda n(l+k)}\biggr]
\end{align*}
\begin{align*}
+ \frac{\lambda}{2\pi (j\lambda n)^{p+1}}\sum\limits_{k\in \mathcal K_\bullet}\biggl[P_k\sum\limits_{\ell=0}^{p+1}(-1)^\ell{p+1\choose \ell}e^{j\lambda n(l+k)}\biggr]  \\
= \frac{\lambda}{2\pi (j\lambda n)^{p+1}}\sum\limits_{k\in \mathcal K_\bullet}\biggl[P_k\sum\limits_{\ell=0}^{p+1}(-1)^\ell{p+1\choose \ell} \\ \biggl(\cos\lambda n(l+k)-j\sin\lambda n(l+k)\biggr)\biggr]  \\ + \frac{\lambda}{2\pi (j\lambda n)^{p+1}}\sum\limits_{k\in \mathcal K_\bullet}\biggl[P_k\sum\limits_{\ell=0}^{p+1}(-1)^\ell{p+1\choose \ell} \\ \biggl(\cos\lambda n(l+k)+j\sin\lambda n(l+k)\biggr)\biggr]  \\
= \frac{\lambda}{2\pi (j\lambda n)^{p+1}}\sum\limits_{k\in \mathcal K_\bullet}\biggl[P_k\sum\limits_{\ell=0}^{p+1}(-1)^\ell{p+1\choose \ell}\cos\lambda n(l+k) \\-(-1)^\ell{p+1\choose \ell}j\sin\lambda n(l+k)\biggr] \\
+ \frac{\lambda}{2\pi (j\lambda n)^{p+1}}\sum\limits_{k\in \mathcal K_\bullet}\biggl[P_k\sum\limits_{\ell=0}^{p+1}(-1)^{p+1-\ell}{p+1\choose \ell}\cos\lambda n(l+k) \\+(-1)^{p+1-\ell}{p+1\choose \ell}j\sin\lambda n(l+k)\biggr] \\
    = \frac{\lambda}{2\pi (j\lambda n)^{p+1}}\sum\limits_{k\in \mathcal K_\bullet}\biggl[P_k\sum\limits_{\ell=0}^{p+1}(-1)^\ell{p+1\choose \ell}\biggl[\cos\lambda n(l+k)  \\ 
    +(-1)^{p+1}\cos\lambda n(l+k)-j\sin\lambda n(l+k) \\+(-1)^{p+1}j\sin\lambda n(l+k)\biggr]\biggr]  \\
    = \frac{\lambda}{2\pi (j\lambda n)^{p+1}}\sum\limits_{k\in \mathcal K_\bullet}\biggl[P_k\sum\limits_{\ell=0}^{p+1}(-1)^\ell{p+1\choose \ell}\biggl[\bigl[1 \\+(-1)^{p+1}\bigr]\cos\lambda n(l+k)-j\bigl[1-(-1)^{p+1}\bigr]\sin\lambda n(l+k)\biggr]\biggr].
\end{align*}
Therefore, 
\begin{align*}
    A_n=\frac{\lambda}{2\pi (j\lambda n)^{p+1}}\sum\limits_{k\in \mathcal K_\bullet}P_k\sum\limits_{\ell=0}^{p+1}(-1)^\ell{p+1\choose \ell}E_{n,\ell,k},
\end{align*}
where 
\begin{align*}
    E_{n,\ell,k}=\bigl[1+(-1)^{p+1}\bigr]cos\lambda n(l+k) \\ 
    -j\bigl[1-(-1)^{p+1}\bigr]sin\lambda n(l+k).
\end{align*}
Next, for the $B_n$ coefficient we have that
\begin{align*}
    B_n=\frac{A_n-2D_n}{j} = -j(A_n-2Dn)
\end{align*}
Consequently,
\begin{align*}
    B_n=-j\biggl(\frac{\lambda}{2\pi (j\lambda n)^{p+1}}\sum\limits_{k\in \mathcal K_\bullet}P_k\sum\limits_{\ell=0}^{p+1}(-1)^\ell{p+1\choose \ell}\\\biggl[\bigl[1+(-1)^{p+1}\bigr]\cos\lambda n(l+k) \\ 
    -j\bigl[1-(-1)^{p+1}\bigr]\sin\lambda n(l+k)\biggr] \\ -\frac{2}{2\pi}\sum\limits_{k\in \mathcal K_\bullet}\left[P_k\left(\frac{1-e^{-j\lambda n}}{j\lambda n}\right)^{p+1}\cdot \lambda e^{-j\lambda n k}\right]\biggr) \\
    = -j\biggl(\frac{\lambda}{2\pi (j\lambda n)^{p+1}}\sum\limits_{k\in \mathcal K_\bullet}P_k\sum\limits_{\ell=0}^{p+1}(-1)^\ell{p+1\choose \ell}\\\biggl[\bigl[1+(-1)^{p+1}\bigr]\cos\lambda n(l+k) \\ 
    -j\bigl[1-(-1)^{p+1}\bigr]\sin\lambda n(l+k)\biggr] \\
    -\frac{\lambda}{\pi (j\lambda n)^{p+1}}\sum\limits_{k\in \mathcal K_\bullet}P_k\sum\limits_{\ell=0}^{p+1}(-1)^\ell{p+1\choose \ell}e^{-j\lambda n\ell}e^{-j\lambda nk}\biggr) \\
    = -j\biggl(\frac{\lambda}{2\pi (j\lambda n)^{p+1}}\sum\limits_{k\in \mathcal K_\bullet}P_k\sum\limits_{\ell=0}^{p+1}(-1)^\ell{p+1\choose \ell}\\\biggl[\bigl[1+(-1)^{p+1}\bigr]\cos\lambda n(l+k) \\ 
    -j\bigl[1-(-1)^{p+1}\bigr]\sin\lambda n(l+k)\biggr] \\
    -\frac{\lambda}{\pi (j\lambda n)^{p+1}}\sum\limits_{k\in \mathcal K_\bullet}P_k\sum\limits_{\ell=0}^{p+1}(-1)^\ell{p+1\choose \ell}e^{-j\lambda n(\ell+k)}\biggr)\\
    = -j\biggl(\frac{\lambda}{2\pi (j\lambda n)^{p+1}}\sum\limits_{k\in \mathcal K_\bullet}P_k\sum\limits_{\ell=0}^{p+1}(-1)^\ell{p+1\choose \ell}\\\biggl[\bigl[1+(-1)^{p+1}\bigr]\cos\lambda n(l+k) \\ 
    -j\bigl[1-(-1)^{p+1}\bigr]\sin\lambda n(l+k)\biggr] \\
    -\frac{\lambda}{\pi (j\lambda n)^{p+1}}\sum\limits_{k\in \mathcal K_\bullet}P_k\sum\limits_{\ell=0}^{p+1}(-1)^\ell{p+1\choose \ell}\\\bigl[\cos\lambda n(l+k)-j\sin\lambda n(l+k)\bigr]\biggr)\\
    =\biggl(\frac{\lambda}{2\pi (j\lambda n)^{p+1}}\sum\limits_{k\in \mathcal K_\bullet}P_k\sum\limits_{\ell=0}^{p+1}(-1)^\ell{p+1\choose \ell} \\
    \cdot(-j)\bigl[[1+(-1)^{p+1}]\cos\lambda n(l+k) \\
    -j[1-(-1)^{p+1}\sin\lambda n(l+k)]\bigr] \\
    + \frac{\lambda}{2\pi (j\lambda n)^{p+1}}\sum\limits_{k\in \mathcal K_\bullet}P_k\sum\limits_{\ell=0}^{p+1}(-1)^\ell{p+1\choose \ell} \\
    \cdot (2j)[\cos\lambda n(l+k)-j\sin\lambda n(l+k)]\biggr)
\end{align*}
\begin{align*}
    = \frac{\lambda}{2\pi (j\lambda n)^{p+1}}\sum\limits_{k\in \mathcal K_\bullet}P_k\sum\limits_{\ell=0}^{p+1}(-1)^\ell{p+1\choose \ell} \\
    \biggl[-j\cos\lambda n(l+k)[1+(-1)^{p+1}-2] \\
    -[1-(-1)^{p+1}-2]\sin\lambda n(l+k)]\biggr].
\end{align*}
\vspace{-2.2em}
Therefore, \[
    B_n=\frac{\lambda}{2\pi (j\lambda n)^{p+1}}\sum\limits_{k\in \mathcal K_\bullet}P_k\sum\limits_{\ell=0}^{p+1}(-1)^\ell{p+1\choose \ell}F_{n,\ell,k},
\] where \[
    F_{n,\ell,k}=j[1-(-1)^{p+1}]\cos\lambda n(l+k)\\
    + [1+(-1)^{p+1}]\sin\lambda n(l+k).
\]
\hfill\hfill \qed
\end{pf}
\vspace{-2.55em}
For the next result we recall the definition of the polylogarithm $\li_{p+1}(z)$ of order $p+1$,
\begin{equation}
\label{eq:polylog}
 \li_{p+1}(z) = \sum_{n=1}^\infty \frac{z^n}{n^{p+1}},
\end{equation}
\vspace{-.5em}
which allows to express $\gamma(r,\theta)$ as stated next.
\begin{proposition}
\label{prop:gamma_cardinal_spline}
With Fourier coefficients \eqref{eq:fourier_coefficients}, \eqref{eq:harmonic_fourier} becomes
\begin{align}
\label{eq:harmonic_potential_bspline}
    &\mkern-12mu\gamma(r,\theta)=\frac{\lambda}{2\pi}\sum\limits_{k\in \mathcal K_\bullet}P_k\biggl[1+\frac{1}{(j\lambda)^{p+1}} \sum\limits_{\ell=0}^{p+1}\varrho(p+1,\ell)   \\
    &\nonumber \mkern-12mu\biggl(\mkern-4mu(-1)^{p+1}\mkern-4mu\li_{p+1}\bigl(re^{j\left[\lambda(\ell+k)-\theta\right]}\bigr)
    \mkern-4mu+\mkern-4mu\li_{p+1}\bigl(re^{-j\left[\lambda(\ell+k)-\theta\right]}\bigr)\biggr)\mkern-4mu\biggr].
\end{align}
\end{proposition}
\vspace{-2.5em}
\begin{pf}
Replacing \eqref{eq:fourier_coefficients} into \eqref{eq:harmonic_fourier}, we arrive at
\begin{multline}
\label{eq:harmonic_potential_bspline1}
    \gamma(r,\theta)=A_0+\sum\limits_{n\geq 1} r^n\left(A_n\cos n\theta + B_n\sin n\theta\right)\\
    =\frac{\lambda}{2\pi}\sum\limits_{k\in \mathcal K_\bullet}P_k +\sum\limits_{n\geq 1} r^n\biggl(\frac{\lambda}{2\pi(j\lambda n)^{p+1}}\sum\limits_{k\in \mathcal K_\bullet}P_k \\ \sum\limits_{\ell=0}^{p+1}(-1)^\ell{p+1\choose \ell}E_{n,\ell,k}\cos n\theta \\ 
    +\frac{\lambda}{2\pi(j\lambda n)^{p+1}}\sum\limits_{k\in \mathcal K_\bullet}P_k\sum\limits_{\ell=0}^{p+1}(-1)^\ell{p+1\choose \ell}F_{n,\ell,k}\sin n\theta\biggr) \\
    =\frac{\lambda}{2\pi}\sum\limits_{k\in \mathcal K_\bullet}P_k\sum\limits_{n\geq 1} r^n\biggl[\frac{\lambda}{2\pi(j\lambda n)^{p+1}}\sum\limits_{k\in \mathcal K_\bullet}P_k\sum\limits_{\ell=0}^{p+1}(-1)^\ell{p+1\choose \ell}\\
    \biggl(E_{n,\ell,k}\cos n\theta+F_{n,\ell,k}\sin n\theta\biggr)\biggr]\\
    =\frac{\lambda}{2\pi}\sum\limits_{k\in \mathcal K_\bullet}P_k\biggl[1+\frac{1}{(j\lambda)^{p+1}}    \sum\limits_{n\geq 1}\frac{r^n}{n^{p+1}}\\
    \sum\limits_{\ell=0}^{p+1}(-1)^\ell{p+1\choose \ell}L_{n,\ell,k, \theta}\biggr],
\end{multline}
with shorthand notation $L_{n,\ell,k, \theta}$ given by
\begin{multline}
\label{eq:L_notation}
    L_{n,\ell,k, \theta}=E_{n,\ell,k}\cos n\theta +F_{n,\ell,k}\sin n\theta \\
    =\cos n\theta\biggl[\bigl[1+(-1)^{p+1}\bigr]\cos \lambda n(l+k)\\
    -j\bigl[1-(-1)^{p+1}\sin\lambda n(l+k)\bigr]\biggr] \\
    +\sin n\theta\biggl[j\bigl[1-(-1)^{p+1}\bigr]\cos\lambda n(l+k) \\
    +\bigl[1+(-1)^{p+1}\bigr]\sin\lambda n(l+k)\biggr] \\
    =\cos n\theta\cos\lambda n(l+k)\bigl[1+(-1)^{p+1}\bigr] \\
    -j\cos n\theta\sin\lambda n(l+k)\bigl[1-(-1)^{p+1}\bigr] \\
    +j\sin n\theta\cos\lambda n(l+k)\bigl[1-(-1)^{p+1}\bigr] \\
    +\sin n\theta\sin\lambda n(l+k)\bigl[1+(-1)^{p+1}\bigr] \\
    =\bigl[1+(-1)^{p+1}\bigr]\cos n\bigl[\lambda(\ell+k)-\theta\bigr]\\
    -j\bigl[1-(-1)^{p+1}\bigr]\sin n\bigl[\lambda(\ell+k)-\theta\bigr].
\end{multline}
Recalling that $r^n\cos n\phi$ can be written as $[(re^{j\phi})^n+(re^{-j\phi})^n]/2$ and $r^n\sin n\phi$ as $[(re^{j\phi})^n-(re^{-j\phi})^n]/2j$ for $\phi = \lambda(\ell + k) - \theta$, we use \eqref{eq:polylog} and \eqref{eq:L_notation} to obtain
\begin{multline*}
    \sum\limits_{n\geq 1}\frac{r^n}{n^{p+1}}L_{n,\ell,k, \theta}=\\
    \sum\limits_{n\geq 1}\frac{r^n}{n^{p+1}}\bigl[1+(-1)^{p+1}\bigr]\cos n[\lambda(\ell+k)-\theta]\\
    -\sum\limits_{n\geq 1}\frac{r^n}{n^{p+1}}j\bigl[1-(-1)^{p+1}\bigr]\sin n[\lambda(\ell+k)-\theta]\\
     =\sum\limits_{n\geq 1}\frac{r^n}{n^{p+1}}\cos n[\lambda(\ell+k)-\theta] \\
    +\sum\limits_{n\geq 1}\frac{r^n}{n^{p+1}}(-1)^{p+1}\cos n[\lambda(\ell+k)-\theta]\\
    -\sum\limits_{n\geq 1}j\frac{r^n}{n^{p+1}}\sin n[\lambda(\ell+k)-\theta] \\ 
    +\sum\limits_{n\geq 1}j\frac{r^n}{n^{p+1}}(-1)^{p+1}\sin n[\lambda(\ell+k)-\theta]\\
    =\sum\limits_{n\geq 1}\frac{r^n}{n^{p+1}}\bigl[\cos n[\lambda(\ell+k)-\theta]-j\sin n[\lambda(\ell+k)-\theta]\bigr]\\
    +(-1)^{p+1}\sum\limits_{n\geq 1}\frac{r^n}{n^{p+1}}\bigl[\cos n[\lambda(\ell+k)-\theta]+j\sin n[\lambda(\ell+k)-\theta]\bigr] \\
    =\sum\limits_{n\geq 1}\frac{r^n}{n^{p+1}}e^{-jn[\lambda(\ell+k)-\theta]} + (-1)^{p+1}\sum\limits_{n\geq 1}\frac{r^n}{n^{p+1}}e^{jn[\lambda(l+k)-\theta]}\\
    =\sum\limits_{n\geq 1}\frac{\bigl(re^{-j[\lambda(\ell+k)-\theta]}\bigr)^n}{n^{p+1}}+(-1)^{p+1}\sum\limits_{n\geq 1}\frac{\bigl(re^{j[\lambda(\ell+k)-\theta]}\bigr)^n}{n^{p+1}}
\end{multline*}
\begin{multline}
\label{eq:harmonic_potential_bspline2}
    =(-1)^{p+1}\mkern-4mu\li_{p+1}\bigl(re^{j\left[\lambda(\ell+k)-\theta\right]}\bigr)
    \mkern-4mu+\mkern-4mu\li_{p+1}\bigl(re^{-j\left[\lambda(\ell+k)-\theta\right]}\bigr).
\end{multline}
Introducing \eqref{eq:harmonic_potential_bspline2} in \eqref{eq:harmonic_potential_bspline1} leads to \eqref{eq:harmonic_potential_bspline}, thus concluding the proof. \hfill\hfill\qed
\end{pf}
\begin{remark}
 Unfortunately, for any $p > 0$, \eqref{eq:polylog} defines a transcendental function, i.e., one which cannot be expressed in terms of elementary functions. Thus, for practical implementations we resort to approximations of \eqref{eq:polylog}. The simplest approach is to truncate the defining series, but more refined methods based on the zeta function approximations and alternative series representations are also available~\cite[eq.~(27)]{bhagat2003evaluation}. \eor
\end{remark}
\subsection{Choosing the control points}

To construct the harmonic potential in \eqref{eq:harmonic_potential_bspline}, we must appropriately choose the control points $P_k$ defining the boundary condition \eqref{eq:dirichlet_spline}. The next lemma assists in this.
\begin{lemma}
\label{lem:pk}
For scalars $\theta_0,\theta_1 \in (0, 2\pi)$ with $\theta_0 < \theta_1$, take
\begin{equation}
    \label{eq:pk}
    P_k = \begin{cases}
          -1, &  k\in \mathcal K_\theta,\\
          \hphantom{-}0, & \text{otherwise}. 
        \end{cases}  
\end{equation}
where\vspace{-2em} 
\begin{equation}
\label{eq:index_theta}
    \mathcal K_{\theta} = \biggl\{\left\lceil \frac{\theta_0}{\lambda}\right\rceil\leq k\leq \left\lfloor \frac{\theta_1}{\lambda}\right\rfloor-p-1\biggr\},
\end{equation}
and\vspace{-1em}
\begin{equation}
\label{eq:lambda_condition}
    \lambda < \frac{\theta_1-\theta_0}{2(p+1)}.
\end{equation}
Under these conditions, there exist
\begin{equation}
\label{eq:theta_prime}
    \theta_0'=\lambda(\underline k +p),\quad \theta_1'=\lambda \overline k,
\end{equation}
such that $\theta_0 \leq \theta_0' < \theta_1' \leq \theta_1$, for which $h(\theta)$ satisfies:
\begin{enumerate}[label=\roman*)]
    \item\label{item:req-i} $h(\theta)=-1$, for all $\theta \in [\theta_0',\theta_1']$;
    \item\label{item:req-ii} $h(\theta)$ has no local minima outside $[\theta_0',\theta_1']$.
\end{enumerate}
For later use, we denote
\begin{equation}    
\label{eq:k_bar_notation}
    \underline k = \left\lceil \frac{\theta_0}{\lambda}\right\rceil,
    \quad
    \overline k = \left\lfloor \frac{\theta_1}{\lambda}\right\rfloor - p.
\end{equation}
\end{lemma}
\vspace{-3em}
\begin{pf}
Recalling \eqref{eq:sigma_support} we have that all indices for which the support of $\sigma_k(\theta)$ is fully inside the interval $[\theta_0, \theta_1]$ are those which respect $\frac{\theta_0}{\lambda}\leq k\leq \frac{\theta_1}{\lambda}-p-1$, which may be equivalently stated as in \eqref{eq:index_theta}, where we used the properties $\lfloor x\pm n\rfloor=\lfloor x\rfloor\pm n$, $\lceil x\pm n\rceil=\lceil x\rceil\pm n$. Introducing \eqref{eq:pk} into \eqref{eq:dirichlet_spline} and applying $d/d\theta$, we obtain
\vspace{-1em}
\begin{align}
    \frac{d}{d\theta}h(\theta)&=-\frac{d}{d\theta}\sum\limits_{k=\underline k}^{\overline k-1}\sigma_k(\theta)\\
    \nonumber &\overset{\ref{prop:derivative}}{=}-\frac{1}{\lambda}\sum\limits_{k=\underline k}^{\overline k-1}\biggl[M_{p-1}\biggl(\frac{\theta}{\lambda} - k\biggr)-M_{p-1}\biggl(\frac{\theta}{\lambda} - k-1\biggr)\biggr]\\ 
    \nonumber &=-\frac{1}{\lambda}\biggl[M_{p-1}\biggl(\frac{\theta}{\lambda}-\underline k\biggr)-M_{p-1}\biggl(\frac{\theta}{\lambda}-\underline k-1\biggr)\\
    \nonumber &+M_{p-1}\biggl(\frac{\theta}{\lambda}-\underline k-1\biggr)-...-M_{p-1}\biggl(\frac{\theta}{\lambda}-\overline k\biggr)\biggr]\\
    \nonumber&=\frac{1}{\lambda}\biggl[-M_{p-1}\biggl(\frac{\theta}{\lambda} - \underline k\biggr)+M_{p-1}\biggl(\frac{\theta}{\lambda} - \overline k\biggr)\biggr].
\end{align}
Replacing $p$ with $p-1$ in \eqref{eq:sigma_support} and since the spline functions $M_{p-1}(\cdot)$ are positive, we obtain the sign table:
\begin{center}
\includegraphics[width=\columnwidth]{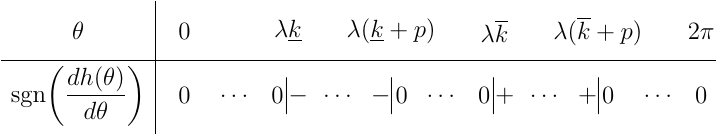}
    \label{fig:sign-table}
\end{center}
\vspace{-2em}

The underlying assumption in constructing this table is that the orderings $0<\lambda \underline k< \lambda(\underline k+p)<\lambda \overline k<\lambda (\overline k+p)<2\pi$ hold. Using the prior notation for $\underline k,\overline k$ in \eqref{eq:k_bar_notation} and removing redundancies, these conditions reduce to $0<\underline k, \quad \overline k\leq \left\lceil\frac{2\pi}{\lambda}\right\rceil-p, \quad \underline k+p<\overline k$, which, together with the properties $\lceil x\rceil \leq x + 1$ and $\lfloor x\rfloor \geq x - 1$, lead to the sufficient condition \eqref{eq:lambda_condition}.  Furthermore, since $h(0) = h(2\pi) = 0$ by construction, the sign table shows that the interval on which $h(\theta) = -1$ is precisely the one in item~\ref{item:req-i}, for $\theta_0',\theta_1'$ given in \eqref{eq:theta_prime}. Finally, the sign pattern also shows that on this same interval the function attains its minimum and that, outside it, there is no other local minimum, verifying item~\ref{item:req-ii}.\hfill\hfill\qed
\end{pf}
\vspace{-2em}
Using the previous result we arrive at the final formulation for the harmonic potential surface, one which considers the choice of control points from Lemma~\ref{lem:pk}.
\begin{corollary}
\label{cor:harmonic_surface}
With $h(\theta)$ from \eqref{eq:dirichlet_spline} verifying \eqref{eq:pk}--\eqref{eq:theta_prime}, the harmonic surface \eqref{eq:harmonic_potential_bspline} becomes
\vspace{-.5em}
\begin{align}
    \label{eq:harmonic_potential_bspline3final}
    &\gamma(r,\theta)=-\frac{\lambda(\overline k - \underline k)}{2\pi}-\frac{\lambda}{2\pi(j\lambda)^{p+1}}\sum\limits_{\ell=0}^{p}\varrho(p,\ell)\\
    &\nonumber\bigl[
    \mkern-2mu (-1)^{p+1}\biggl(\mkern-3mu \li_{p+1}\mkern-4mu\left(re^{j[\lambda (\ell+\underline k)-\theta]}\right)\mkern-4mu
    - \mkern-4mu\li_{p+1}\mkern-4mu\left(re^{j[\lambda (\ell+\overline k)-\theta]}\right)\mkern-5mu\biggr)\\
    &\nonumber+\li_{p+1}\mkern-4mu\left(re^{-j[\lambda (\ell+\underline k)-\theta]}\right)\mkern-4mu
    - \mkern-4mu\li_{p+1}\mkern-4mu\left(re^{-j[\lambda (\ell+\overline k)-\theta]}\right)\biggr].
\end{align}
\end{corollary}
\vspace{-3em}
\begin{pf}
Introducing into \eqref{eq:harmonic_potential_bspline1} the control points \eqref{eq:pk} for the index set $\mathcal K_\theta$ defined in \eqref{eq:index_theta} under condition \eqref{eq:lambda_condition}, and reordering the sums, we arrive at:
\vspace{-.5em}
\begin{multline}
    \label{eq:harmonic_potential_bspline3}
    \gamma(r,\theta)=-\frac{\lambda}{2\pi}\biggl[(\overline k - \underline k) +\frac{1}{(j\lambda)^{p+1}}\sum\limits_{\ell=0}^{p+1}\varrho(p+1,\ell)\\
    \sum\limits_{n \geq 1}\frac{r^n}{n^{p+1}}\biggl(\sum\limits_{k=\underline k}^{\overline k-1}L_{n,\ell,k,\theta}\biggr)\biggr].
\end{multline}
From theory we have that:
\begin{align*}
    \sum\limits_{k=0}^{N-1} e^{jkx}=\frac{e^{j0x}-e^{jNx}}{1-e^{jx}}=\frac{1-e^{jNx}}{1-e^{jx}},
\end{align*}
which leads to:
\begin{align*}
    \sum\limits_{k=\underline k}^{\overline{k}-1} e^{jn[\lambda(l+k)-\theta]}=\sum\limits_{k=\underline k}^{\overline{k}-1} e^{jn\lambda \ell} e^{jn\lambda k} e^{-jn\theta} = e^{jn(\lambda \ell-\theta)}\sum\limits_{k=\underline k}^{\overline{k}-1} e^{jn\lambda k} \\
    = e^{jn(\lambda \ell -\theta)}\frac{e^{jn\lambda \underline k} - e^{jn\lambda \overline{k}}}{1-e^{jn\lambda}} = \frac{e^{jn[\lambda(\ell+\underline k)-\theta]}-e^{jn[\lambda(\ell+\overline{k})-\theta]}}{1-e^{jn\lambda}}.
\end{align*}
For compactness, we note $\phi=n\lambda$, $\psi=n(\lambda \ell-\theta)$ and obtain
\begin{align*}
    \sum\limits_{k=\underline k}^{\overline{k}-1} e^{jn[\lambda(l+k)-\theta]}=\sum\limits_{k=\underline k}^{\overline{k}-1} e^{j(\phi k+\psi)}=\frac{e^{j(\phi \underline k+\psi)}-e^{j(\phi \overline{k}+\psi)}}{1-e^{j\phi}},  
\end{align*}
which further allows one to write
\begin{multline*}
    \sum\limits_{k=\underline k}^{\overline k-1}L_{n,\ell,k,\theta}=\sum\limits_{k=\underline k}^{\overline k-1}(-1)^{p+1}e^{jn[\lambda(\ell+k)-\theta]}+e^{-jn[\lambda(\ell+k)-\theta]}\\
    =\sum\limits_{k=\underline k}^{\overline k-1}(-1)^{p+1}e^{j(\phi k+\psi)}+e^{-j(\phi k+\psi)}\\
    =(-1)^{p+1}\sum\limits_{k=\underline k}^{\overline k-1}e^{j(\phi k+\psi)}+\sum\limits_{k=\underline k}^{\overline k-1}e^{-j(\phi k+\psi)}\\
    =(-1)^{p+1}\underbrace{\frac{e^{j(\phi \underline k+\psi)}-e^{j(\phi \overline k+\psi)}}{1-e^{j\phi}}}_{(\star)}+\underbrace{\frac{e^{-j(\phi \underline k+\psi)}-e^{-j(\phi \overline k+\psi)}}{1-e^{-j\phi}}}_{(\star\star)}
\end{multline*}
Tedious but straightforward calculations render term 
\begin{align}
\nonumber &\sum\limits_{\ell=0}^{p+1}(-1)^\ell{p+1\choose \ell}
    \sum\limits_{n \geq 1}\frac{r^n}{n^{p+1}}\cdot (\star)=\\
    \nonumber&=\sum\limits_{\ell=0}^{p+1}(-1)^\ell{p+1\choose \ell}
    \sum\limits_{n \geq 1}\frac{r^n}{n^{p+1}}\frac{e^{jn[\lambda(\ell+\underline k)-\theta]}-e^{jn[\lambda(\ell+\overline{k})-\theta]}}{1-e^{jn\lambda}}\\
\nonumber&=\sum\limits_{\ell=0}^{p+1}(-1)^\ell{p+1\choose \ell}
    \sum\limits_{n\geq 1}\frac{r^n}{n^{p+1}}\frac{e^{jn(\lambda \ell-\theta)}}{1-e^{jn\lambda}}\biggl[e^{jn\lambda\underline k}-e^{jn\lambda\overline k}\biggr]\\
\nonumber     &=\sum\limits_{n\geq 1}\frac{r^n}{n^{p+1}}\frac{e^{-jn\theta}}{1-e^{jn\lambda}} \biggl[e^{jn\lambda\underline k}-e^{jn\lambda\overline k}\biggr]\sum\limits_{\ell=0}^{p+1}e^{jn\lambda \ell}(-1)^\ell{p+1\choose \ell}\\
\nonumber     &=\sum\limits_{n\geq 1}\frac{r^n}{n^{p+1}}\frac{e^{-jn\theta}}{1-e^{jn\lambda}} \biggl[e^{jn\lambda\underline k}-e^{jn\lambda\overline k}\biggr]\left(1-e^{jn\lambda}\right)^{p+1}\\
\label{eq:harmonic_potential_bspline4}    &=\sum\limits_{n\geq 1}\frac{r^ne^{-jn\theta}}{n^{p+1}} \biggl[e^{jn\lambda\underline k}-e^{jn\lambda\overline k}\biggr]\left(1-e^{jn\lambda}\right)^{p}\\
\nonumber     &=\sum\limits_{\ell=0}^{p}(-1)^\ell{p\choose \ell}\sum\limits_{n\geq 1}\frac{r^ne^{-jn\theta}}{n^{p+1}} \biggl[e^{jn\lambda\underline k}-e^{jn\lambda\overline k}\biggr]e^{jn\lambda \ell}\\
\nonumber     &=\sum\limits_{\ell=0}^{p}(-1)^\ell{p\choose \ell}\sum\limits_{n\geq 1}\frac{r^n}{n^{p+1}} \biggl[e^{jn[\lambda(\ell+\underline k)-\theta]}-e^{jn[\lambda(\ell+\overline k)-\theta]}\biggr]\\
\nonumber   &=\sum\limits_{\ell=0}^{p}(-1)^\ell{p\choose \ell}\biggl(\sum\limits_{n\geq 1}\frac{r^n}{n^{p+1}}e^{jn[\lambda(\ell+\underline k)-\theta]}\\
\nonumber &-\sum\limits_{n\geq 1}\frac{r^n}{n^{p+1}}e^{jn[\lambda(\ell+\overline{k})-\theta]}\biggr)\\
\nonumber &=\sum\limits_{\ell=0}^{p}(-1)^\ell{p\choose \ell}\mkern-4mu\biggl(\sum\limits_{n\geq 1}\frac{\bigl(re^{j[\lambda(\ell+\underline k)-\theta]}\bigr)^n}{n^{p+1}}\mkern-5mu-\mkern-7mu \sum\limits_{n\geq 1}\frac{\bigl(re^{j[\lambda(\ell+ \overline{k})-\theta]}\bigr)^n}{n^{p+1}}\mkern-3mu\biggr)\\
\nonumber     &=\mkern-6mu\sum\limits_{\ell=0}^{p}(-1)^\ell{p\choose \ell}\mkern-6mu\left[\li_{p+1}\mkern-4mu\left(re^{j[\lambda (\ell+\underline k)-\theta]}\right)\mkern-2mu-\mkern-2mu\li_{p+1}\mkern-4mu\left(re^{j[\lambda (\ell+\overline k)-\theta]}\right)\right].
\end{align}
Further noting that $(\star\star)$ is the conjugate of $(\star)$ gives
\begin{multline*}
    \sum\limits_{\ell=0}^{p+1}(-1)^\ell{p+1\choose \ell}
    \sum\limits_{n \geq 1}\frac{r^n}{n^{p+1}}\cdot (\star\star)\\
    =\sum\limits_{\ell=0}^{p+1}(-1)^\ell{p+1\choose \ell}\mkern-4mu
    \sum\limits_{n \geq 1}\frac{r^n}{n^{p+1}}\frac{e^{-jn[\lambda(\ell+\underline k)-\theta]}-e^{-jn[\lambda(\ell+\overline{k})-\theta]}}{1-e^{jn\lambda}}\\
    =\sum\limits_{\ell=0}^{p+1}(-1)^\ell{p+1\choose \ell}\sum\limits_{n \geq 1}\frac{r^n}{n^{p+1}}\frac{e^{-jn(\lambda \ell-\theta)}}{1-e^{-jn\lambda}}\biggl[e^{-jn\lambda\underline k}-e^{-jn\lambda\overline k}\biggr]\\
    =\sum\limits_{n\geq 1}\frac{r^n}{n^{p+1}}\frac{e^{jn\theta}}{1-e^{-jn\lambda}} \biggl[e^{-jn\lambda\underline k}-e^{-jn\lambda\overline k}\biggr]\\
    \cdot\sum\limits_{\ell=0}^{p+1}e^{-jn\lambda \ell}(-1)^\ell{p+1\choose \ell}
\end{multline*}
\begin{multline}
\label{eq:harmonic_potential_bspline5}
    =\sum\limits_{n\geq 1}\frac{r^n}{n^{p+1}}\frac{e^{jn\theta}}{1-e^{-jn\lambda}} \biggl[e^{-jn\lambda\underline k}-e^{-jn\lambda\overline k}\biggr](1-e^{-jn\lambda})^{p+1}\\
    =\sum\limits_{n\geq 1}\frac{r^ne^{jn\theta}}{n^{p+1}} \biggl[e^{-jn\lambda\underline k}-e^{-jn\lambda\overline k}\biggr]\left(1-e^{-jn\lambda}\right)^{p}\\
    \sum\limits_{\ell=0}^{p}(-1)^\ell{p\choose \ell}\sum\limits_{n\geq 1}\frac{r^ne^{jn\theta}}{n^{p+1}} \biggl[e^{-jn\lambda\underline k}-e^{-jn\lambda\overline k}\biggr]e^{-jn\lambda \ell}\\
    =\sum\limits_{\ell=0}^{p}(-1)^\ell{p\choose \ell}\sum\limits_{n\geq 1}\frac{r^n}{n^{p+1}} \biggl[e^{-jn[\lambda(\ell+\underline k)-\theta]}-e^{-jn[\lambda(\ell+\overline k)-\theta]}\biggr]\\
    =\sum\limits_{\ell=0}^{p}(-1)^\ell{p\choose \ell}\biggl(\sum\limits_{n\geq 1}\frac{\bigl(re^{-j[\lambda(\ell+\underline k)-\theta]}\bigr)^n}{n^{p+1}}- \sum\limits_{n\geq 1}\frac{\bigl(re^{-j[\lambda(\ell+ \overline{k})-\theta]}\bigr)^n}{n^{p+1}}\biggr)\\
    =\sum\limits_{\ell=0}^{p}(-1)^\ell{p\choose \ell}\biggl[\li_{p+1}\left(re^{-j[\lambda (\ell+\underline k)-\theta]}\right)\\
    -\li_{p+1}\left(re^{-j[\lambda (\ell+\overline k)-\theta]}\right)\biggr].
\end{multline}
Introducing \eqref{eq:harmonic_potential_bspline4} and \eqref{eq:harmonic_potential_bspline5} back into \eqref{eq:harmonic_potential_bspline3} yields \eqref{eq:harmonic_potential_bspline3final}. 
However, this derivations hold only when $e^{\pm j\phi} \neq 1$. Recalling the general formula for the truncated geometric series
\begin{align*}
    \sum\limits_{k=0}^{N-1} e^{jkx}=\begin{cases}\frac{e^{j0x}-e^{jNx}}{1-e^{jx}}=\frac{1-e^{jNx}}{1-e^{jx}}&, e^{jx}\neq 1,\\
    N &,e^{jx}=1,
    \end{cases}
\end{align*}
we note that when there are values of $n$ for which $n\lambda=2\pi m$ with $m\in \mathbb Z$, we have to consider the second branch of the truncated geometric series. We derive next the computations for this case: 
\begin{align*}
\nonumber &\sum\limits_{k=\underline k}^{\overline{k}-1}e^{j(\phi k + \psi)} = \sum\limits_{k=\underline k}^{\overline{k}-1}e^{j\phi k}e^{j\psi} = \sum\limits_{k=\underline k}^{\overline{k}-1} e^{jn\lambda k}e^{jn(\lambda \ell - \theta)} = \\
\nonumber &=\sum\limits_{k=\underline k}^{\overline{k}-1}\underbrace{e^{jn\lambda(\ell+k)}e^{-jn\theta}}_{n\lambda = 2\pi m, \: m\in \mathbb{Z}} = \sum\limits_{k=\underline k}^{\overline{k}-1}e^{j(\ell+k)2\pi m}e^{-jn\theta}\\
&=(\overline{k} - \underline k)e^{-jn\theta}.
\end{align*}

$\mkern+3mu$ Similarly, $\sum\limits_{k=\underline k}^{\overline{k}-1}e^{-j(\phi k + \psi)} = (\overline{k}-\underline k)e^{jn\theta}.$ Therefore,
\begin{align*} &\sum\limits_{\ell=0}^{p+1}(-1)^\ell{p+1\choose \ell}
    \sum\limits_{n \geq 1}\frac{r^n}{n^{p+1}}\cdot (\star)=\\
    &\sum\limits_{\ell=0}^{p+1}(-1)^\ell{p+1\choose \ell}\sum\limits_{n \geq 1}\frac{r^n}{n^{p+1}}(\overline{k} - \underline k)e^{-jn\theta}=\\
&=(\overline{k} - \underline k)\sum\limits_{n \geq 1}\frac{r^ne^{-jn\theta}}{n^{p+1}}\sum\limits_{\ell=0}^{p+1}(-1)^\ell{p+1\choose \ell}.
\end{align*}

Using Newton's binomial theorem, $(x+y)^{p+1} = \sum_{\ell = 0}^{p+1}{p+1\choose \ell} x^{\ell} y^{p+1-\ell}$, yields:

\begin{align*}
\nonumber &(\overline{k} - \underline k)\sum\limits_{n \geq 1}\frac{r^ne^{-jn\theta}}{n^{p+1}}\sum\limits_{\ell=0}^{p+1}(-1)^\ell{p+1\choose \ell}\\
&=(\overline{k} - \underline k)\sum\limits_{n \geq 1}\frac{r^ne^{-jn\theta}}{n^{p+1}}(-1+1)^{p+1} = 0. 
\end{align*}

The result is the same for both $(\star)$ and $(\star\star)$, therefore, we obtain:

\begin{align*}
    \eqref{eq:harmonic_potential_bspline4}=\begin{cases}\sum\limits_{\ell=0}^{p}(-1)^\ell{p\choose \ell}\\
    \biggl[\li_{p+1}\mkern-4mu\left(re^{j[\lambda (\ell+\underline k)-\theta]}\right)
    -\\ \li_{p+1}\mkern-4mu\left(re^{j[\lambda (\ell+\overline k)-\theta]}\right)\biggr]&, n\lambda \neq 2m\pi ,\\
    0 &,n\lambda = 2m\pi.
    \end{cases}
\end{align*}
\vspace{-1em}
\begin{align*}
    \eqref{eq:harmonic_potential_bspline5}=\begin{cases}\sum\limits_{\ell=0}^{p}(-1)^\ell{p\choose \ell}\\\biggl[\li_{p+1}\left(re^{-j[\lambda (\ell+\underline k)-\theta]}\right)
    -\\ \li_{p+1}\left(re^{-j[\lambda (\ell+\overline k)-\theta]}\right)\biggr]&, n\lambda \neq 2m\pi ,\\
    0 &,n\lambda = 2m\pi.
    \end{cases}
\end{align*}

Fortunately, the singularities do not contribute to the overall term as they vanish.\qed
\end{pf}

\vspace{-1em}

\subsection*{Illustrative example}

To illustrate the construction, we take $\lambda=\pi/20$, $p=3$, and
$T=40$. For $\theta_0=2.52$ and $\theta_1=3.93$, applying
\eqref{eq:pk} yields the boundary function $h(\theta)$ in
\eqref{eq:dirichlet_spline}, shown in Fig.~\ref{fig:spline_curve}. As
desired, $h(\theta)$ varies smoothly and has no local minima outside the
interval $[\theta_0',\theta_1']$, where
$\theta_0'=3.14$ and $\theta_1'=3.45$, as given by
\eqref{eq:theta_prime}. For clarity, the individual cardinal splines are
vertically shifted by `$-1$': those multiplied by $P_k=-1$ are shown in
blue, whereas those multiplied by $P_k=0$ are shown in light gray.

Fig.~\ref{fig:gamma_disk_mapped} depicts the corresponding harmonic
potential \eqref{eq:harmonic_potential_bspline3final} on the unit disk,
with the polylogarithm functions approximated by truncating the series
\eqref{eq:polylog} after the first $15$ terms. Its boundary potential
reproduces the periodic function $h(\theta)$. 

Lastly, consider the polyhedral set:
\begin{equation*}
    P=\biggl\{q\in \mathbb R^2:\:\bbm -0.33& -0.89\\\hphantom{-}0.28& -0.69\\\hphantom{-}0.73& \hphantom{-}0.15\\\hphantom{-}0.48& \hphantom{-}0.83\\-0.58& \hphantom{-}0.65\\-0.87& -0.16\ebm q\geq \bbm \hphantom{-}0.28\\\hphantom{-}0.66\\\hphantom{-}0.65\\\hphantom{-}0.28\\-0.48\\-0.44\ebm\biggr\}. 
\end{equation*}
Using \eqref{eq:betaq}--\eqref{eq:betamax}, we construct $\beta(q)$,
with $\beta_{\max}=0.04\cdot10^{-3}$ and $\beta(q)=0$ on the boundary
of the polytope. Substitution into \eqref{eq:pulled} gives the pull map $\varphi(q)$ from $P$ to the unit
disk. We select as exit the facet with endpoints
$\begin{bsmallmatrix} 0.16 & 0.16 \end{bsmallmatrix}^\top$ and
$\begin{bsmallmatrix} 0.23 & -0.16 \end{bsmallmatrix}^\top$. Mapping these endpoints through
$\varphi(\cdot)$ and expressing the resulting points in polar
coordinates gives $\theta_0=2.52$ and $\theta_1=3.93$, i.e., precisely
the values used in Figs.~\ref{fig:spline_curve}--\ref{fig:gamma_disk_mapped}.
The pulled harmonic potential on the polyhedral cell is then obtained as
$\gamma_P(q)$ through \eqref{eq:gamma_pulled}.

The resulting surface is shown in
Fig.~\ref{fig:polytope_curve_mapped}. On $\partial P$, the pull map sends
the polyhedral boundary to the unit circle, so that $\gamma_P(q)$ assumes
the prescribed boundary values determined by $h(\theta)$. In the
interior, the potential varies smoothly and guides the motion towards the
selected exit facet, shown as the solid red segment. The corresponding
normalized gradient flow is also depicted in the planar polyhedral cell
at the bottom of the figure, illustrating how the vector field directs
trajectories through the prescribed facet while remaining inside the
cell.

\vspace{-1em}

\section{Control policy}
\label{sec:control}
\vspace{-1em}
Consider the double-integrator dynamics
\begin{equation}
\label{eq:dynamics}
    \ddot q=u(q,\dot q).
\end{equation}
For cell $\mathcal P_\ell$, let $\varphi_\ell:\mathcal P_\ell\rightarrow\mathcal C$ denote the cell-specific map \eqref{eq:pulled}, and, following \eqref{eq:gamma_pulled}, define the corresponding pulled-back potential as $\gamma_{P_\ell}(q):=\gamma_\ell(\varphi_\ell(q))$. Following
\cite{conner2003construction}, its normalized negative-gradient field is%
\footnote{
By construction, $\|\nabla_q\gamma_{P_\ell}(q)\|$ does not vanish in the
interior of $\mathcal P_\ell$; see
\cite[Lemma~5.2]{conner2003construction}.}
\begin{equation*}
X_\ell(q)
:=
-\frac{\nabla_q\gamma_{P_\ell}(q)}
       {\|\nabla_q\gamma_{P_\ell}(q)\|}
=
-\frac{
J_{\varphi_\ell}(q)^\top\nabla_\xi\gamma_\ell(\xi)
}{
\|J_{\varphi_\ell}(q)^\top\nabla_\xi\gamma_\ell(\xi)\|
}
\bigg|_{\xi=\varphi_\ell(q)},
\end{equation*}
where $J_{\varphi_\ell}:=\partial\varphi_\ell/\partial q$. Along a
trajectory, $\dot X_\ell=J_{X_\ell}\dot q$.

\begin{figure*}[!ht]
    \centering
    \subfloat[cardinal B-spline curve]{\label{fig:spline_curve}\includegraphics[width=.3\textwidth]{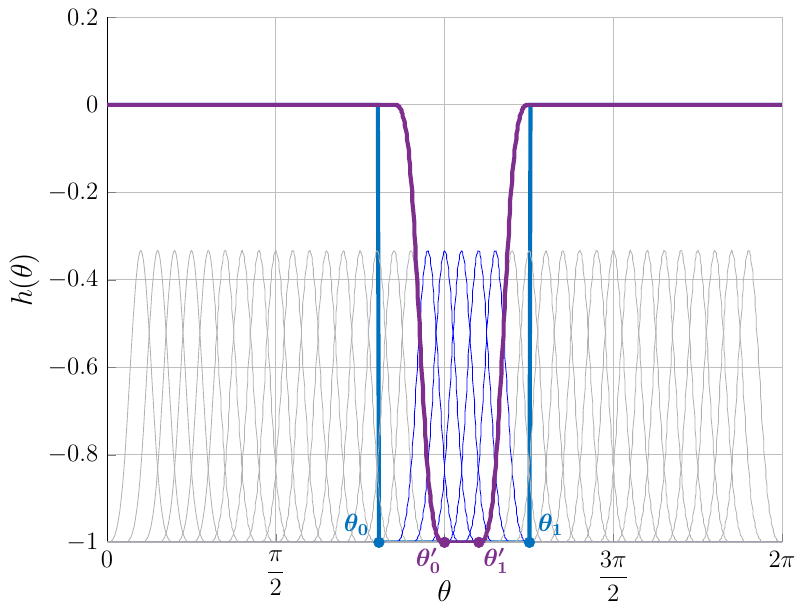}}\hfill
    \subfloat[curve mapped onto the circle and associated potential surface]{\label{fig:gamma_disk_mapped}\includegraphics[width=.3\textwidth]{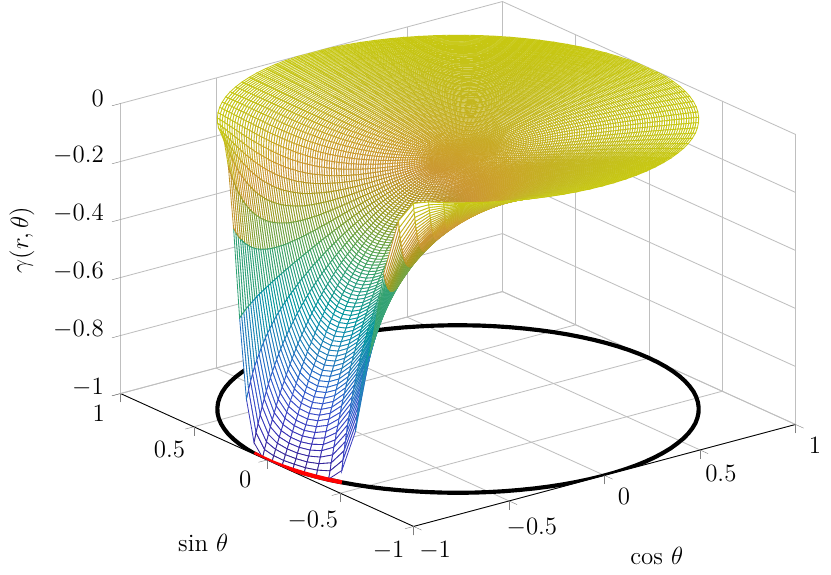}}\hfill
    \subfloat[curve mapped onto the polytope and associated potential surface]{\label{fig:polytope_curve_mapped}\includegraphics[width=.3\textwidth]{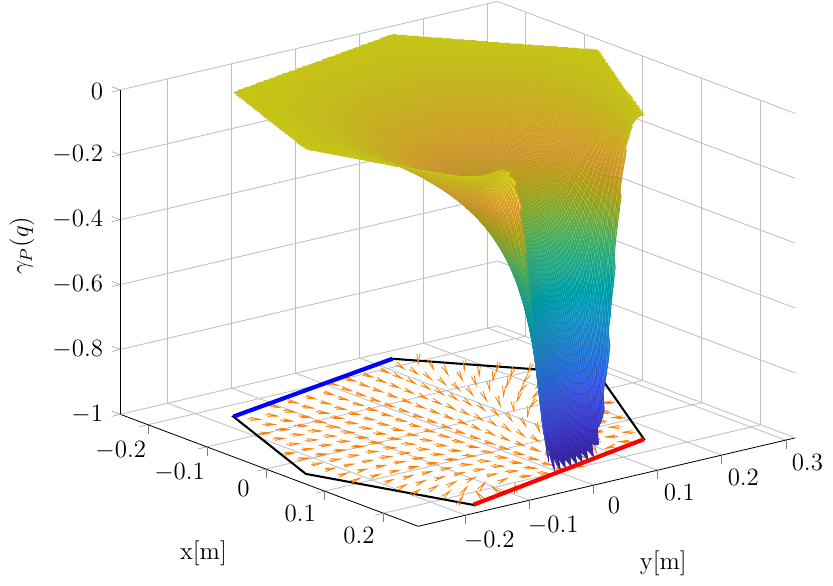}}\hfill
    \caption{Mapping of the induced potential surface from the unit disk to the polyhedral case.}
    \label{fig:harmonic}
\end{figure*}

The ideal single-cell control action is
\begin{equation}
\label{eq:control_action}
    u_\ell(q,\dot q)
    =
    K\bigl(X_\ell(q)-\dot q\bigr)+\dot X_\ell(q),
    \quad K>0.
\end{equation}
The first term regulates the velocity towards $X_\ell$, while the second
compensates for its variation along the trajectory. 
Thus, $\dot q$ converges exponentially to $X_\ell$;
if $\dot q(0)=X_\ell(q(0))$, then $e(t)=X_\ell(q(t))-\dot q(t)=e^{-Kt}e(0)=0,\forall t\geq0$ and the trajectory follows exactly an integral curve of $X_\ell$. 
We next consider continuity of the resulting switched control action.

\subsection{Control action across consecutive cells}
\label{subsec:smooth_control}
\vspace{-1em}

Consider consecutive cells $\mathcal P_{\ell_i}$ and
$\mathcal P_{\ell_{i+1}}$ sharing
$\mathcal F_i=\mathcal P_{\ell_i}\cap\mathcal P_{\ell_{i+1}}$.
Let $\mathcal F_i^\circ\subseteq\mathcal F_i$
denote an open portion of facet $\mathcal F_i$ on which both potentials are constant, and let $n_i$ be the unit normal to $\mathcal F_i$ directed from $\mathcal P_{\ell_i}$ towards $\mathcal P_{\ell_{i+1}}$.

Continuity of \eqref{eq:control_action} across the switching portion
requires
\begin{equation}
\label{eq:generic-interface-requirement}
X_{\ell_i}=X_{\ell_{i+1}},\qquad
\dot X_{\ell_i}=\dot X_{\ell_{i+1}},
\qquad q\in\mathcal F_i^\circ .
\end{equation}
Since $\dot X_\ell=J_{X_\ell}\dot q$, equality
$J_{X_{\ell_i}}=J_{X_{\ell_{i+1}}}$ is sufficient for the second
condition to hold for arbitrary $\dot q$.

\begin{proposition}
\label{prop:interface-matching}
Assume that the gradients do not vanish on $\mathcal F_i^\circ$. If
there exists a constant $\alpha_i>0$ such that
\begin{equation}
\label{eq:interface-gradient-proportionality}
\nabla_q\gamma_{P_{\ell_{i+1}}}(q)
=
\alpha_i\nabla_q\gamma_{P_{\ell_i}}(q),
\qquad q\in\mathcal F_i^\circ,
\end{equation}
then
$X_{\ell_i}=X_{\ell_{i+1}}$ and
$J_{X_{\ell_i}}=J_{X_{\ell_{i+1}}}$ on $\mathcal F_i^\circ$.
Thus, the ideal control action \eqref{eq:control_action} is
continuous on $\mathcal F_i^\circ$ for any $\dot q$.

Furthermore, let
$\mathbf P^{(\ell)}:=\col\bigl(P_k^{(\ell)}\bigr)_{k\in\mathcal K_\bullet}
\in\mathbb R^N$, where $N:=|\mathcal K_\bullet|=T-p$. For prescribed
$\alpha_i>0$, condition \eqref{eq:interface-gradient-proportionality} is
equivalent to a finite homogeneous linear system
\begin{equation}
\label{eq:interface-finite-system}
C_i(\alpha_i)
\begin{bmatrix}
\mathbf P^{(\ell_i)}\\[1mm]
\mathbf P^{(\ell_{i+1})}
\end{bmatrix}
=0,
\:
C_i(\alpha_i)\in\mathbb R^{r_i\times2N},
\:
r_i\leq2N.
\end{equation}
Hence exact interface matching imposes at most $2(T-p)$ independent
linear equality constraints on the control points of two neighboring
cells.
\end{proposition}
\vspace{-1em}

\begin{pf}
Eq.~\eqref{eq:interface-gradient-proportionality} gives
$X:=X_{\ell_i}=X_{\ell_{i+1}}$ and
$\|\nabla_q\gamma_{P_{\ell_{i+1}}}\|
=\alpha_i\|\nabla_q\gamma_{P_{\ell_i}}\|$. Let
$H_\ell:=\nabla_q^2\gamma_{P_\ell}$ and $t_i$ be a unit tangent to
$\mathcal F_i^\circ$. Since $\alpha_i$ is constant, differentiating
\eqref{eq:interface-gradient-proportionality} along $t_i$ gives
$H_{\ell_{i+1}}t_i=\alpha_iH_{\ell_i}t_i$ and, by symmetry of the
Hessians,
$t_i^\top H_{\ell_{i+1}}=\alpha_i t_i^\top H_{\ell_i}$. Moreover,
since the potentials are constant along the facet, their gradients, and
hence $X$, are normal to it. Thus, in $\mathbb R^2$,
$I-XX^\top=t_it_i^\top$, which yields
$(I-XX^\top)H_{\ell_{i+1}}
=\alpha_i(I-XX^\top)H_{\ell_i}$. Since
$J_{X_\ell}
=-\|\nabla_q\gamma_{P_\ell}\|^{-1}
(I-X_\ell X_\ell^\top)H_\ell$, the factor $\alpha_i$ cancels with that
in the gradient norm, yielding
$J_{X_{\ell_{i+1}}}=J_{X_{\ell_i}}$.

Next, we characterize
\eqref{eq:interface-gradient-proportionality} in terms of the control
points. Let $c_\ell$ denote the interior reference point used to
translate $\mathcal P_\ell$ before applying the radial map, and define
$\bar q_\ell=q-c_\ell$. 
Denote by $\beta_\ell$ the function \eqref{eq:betaq} associated with
$\mathcal P_\ell$ and write the cell-specific map \eqref{eq:pulled} as
$\varphi_\ell(q)=r_\ell(q)e_{r,\ell}(q)$, with
$
r_\ell(q)
=
{\|\bar q_\ell\|}/
{\|\bar q_\ell\|+\beta_\ell(q)}$, 
$e_{r,\ell}(q)
=
{\bar q_\ell}/{\|\bar q_\ell\|}
=
\begin{bmatrix}
\cos\theta_\ell(q)&
\sin\theta_\ell(q)
\end{bmatrix}^\top
$. Thus both the radial and angular coordinates may depend on the active
cell. On the common facet,
$\beta_{\ell_i}(q)=\beta_{\ell_{i+1}}(q)=0$, hence
$r_{\ell_i}(q)=r_{\ell_{i+1}}(q)=1$. 
Let $\gamma_\ell(r,\theta)$ denote the harmonic potential on the unit
disk associated with $\mathcal P_\ell$, so that, by
\eqref{eq:gamma_pulled},
$\gamma_{P_\ell}(q)=\gamma_\ell(\varphi_\ell(q))$. We denote its radial
derivative at the boundary by
\begin{multline}
\label{eq:gamma_deriv_r}
\mkern-19mu\gamma_{\ell,r}(1,\mkern-2mu\theta) \mkern-6mu
:= \mkern-10mu
\left.
\frac{\partial\gamma_\ell(r,\mkern-2mu\theta)}{\partial r}
\mkern-4mu\right|_{r=1} \mkern-17mu
= \mkern-8mu
\sum_{n=1}^{\infty}
\mkern-6mu n \mkern-3mu\left(
A_{\ell,n}\mkern-3mu\cos \mkern-1mu n\theta\mkern-5mu+\mkern-5mu
B_{\ell,n}\mkern-3mu\sin \mkern-1mu n\theta
\right)\mkern-3mu.
\end{multline}
Since $\gamma_{P_\ell}$ is constant on $\mathcal F_i^\circ$, its
boundary value is constant on the corresponding arc of the unit circle.
Hence the tangential component of the gradient vanishes on this arc and
the gradient is purely radial:
$
\nabla_\xi\gamma_\ell
\bigl(1,\theta_\ell(q)\bigr)
=
e_{r,\ell}(q)
\gamma_{\ell,r}\bigl(1,\theta_\ell(q)\bigr)$. 
Differentiating the radial map on the boundary gives $J_{\varphi_\ell}(q)^\top e_{r,\ell}(q) = -{\nabla_q\beta_\ell(q)}/\|\bar q_\ell\|
{}$. Consequently, applying the chain rule to
\eqref{eq:gamma_pulled} gives
\begin{equation}
\label{eq:interface-radial-gradient}
\nabla_q\gamma_{P_\ell}(q)
=
-\frac{
\gamma_{\ell,r}(1,\theta_\ell(q))
}{
\|\bar q_\ell\|
}
\nabla_q\beta_\ell(q),
\qquad q\in\mathcal F_i^\circ.
\end{equation}
Define $b_{\ell,i}(q):=n_i^\top\nabla_q\beta_\ell(q)$. Since
$\beta_\ell$ is constant on the facet, $\nabla_q\beta_\ell$ is normal
to it, and hence $\nabla_q\beta_\ell=b_{\ell,i}n_i$. \eqref{eq:interface-radial-gradient} becomes $\nabla_q\gamma_{P_\ell}(q) = -{
b_{\ell,i}(q)\gamma_{\ell,r}(1,\theta_\ell(q))
}/{
\|\bar q_\ell\|
}n_i$. Substituting this expression for the two neighboring cells into
\eqref{eq:interface-gradient-proportionality} and cancelling the common
factor $-n_i$ shows that
\eqref{eq:interface-gradient-proportionality} is equivalent to
\begin{multline}
\label{eq:interface-radial-condition}
\frac{b_{\ell_{i+1},i}(q)}
{\|\bar q_{\ell_{i+1}}\|}
\gamma_{\ell_{i+1},r}
\bigl(1,\theta_{\ell_{i+1}}(q)\bigr)
\\
=
\alpha_i
\frac{b_{\ell_i,i}(q)}
{\|\bar q_{\ell_i}\|}
\gamma_{\ell_i,r}
\bigl(1,\theta_{\ell_i}(q)\bigr),
\qquad q\in\mathcal F_i^\circ.
\end{multline}
Substitution of \eqref{eq:gamma_deriv_r} into
\eqref{eq:interface-radial-condition} gives, for any
$q\in\mathcal F_i^\circ$,
\vspace{-1em}
\begin{align}
\nonumber &\mkern-6mu
\sum_{n=1}^{\infty}\mkern-3mu n \mkern-3mu \Bigg[\mkern-3mu
\frac{b_{\ell_{i+1},i}(q)}
{\|\bar q_{\ell_{i+1}}\|}\mkern-3mu
\Big(\mkern-3mu
A_{\ell_{i+1},n}
\cos n\theta_{\ell_{i+1}}(q)\mkern-5mu
+\mkern-5mu
B_{\ell_{i+1},n}
\sin n\theta_{\ell_{i+1}}(q)\mkern-3mu
\Big)
\label{eq:interface-fourier-condition}
\\[-1mm]
&\mkern+25mu
-\alpha_i
\frac{b_{\ell_i,i}(q)}
{\|\bar q_{\ell_i}\|}\mkern-3mu
\Big(\mkern-3mu
A_{\ell_i,n}
\cos n\theta_{\ell_i}(q) \mkern-5mu
+ \mkern-5mu
B_{\ell_i,n}
\sin n\theta_{\ell_i}(q)\mkern-3mu
\Big)
\mkern-3mu \Bigg] \mkern-6mu= \mkern-2mu 0.
\nonumber
\end{align}
Collecting the control points of $\mathcal P_\ell$ in
$\mathbf P^{(\ell)}$, eq.~\eqref{eq:fourier_coefficients} shows that
each $A_{\ell,n}$ and $B_{\ell,n}$ is linear in
$\mathbf P^{(\ell)}$. Hence, for fixed $\alpha_i$,
we arrive at
\vspace{-1em}
\[
c_i(q)^\top
\begin{bmatrix}
\mathbf P^{(\ell_i)}\\
\mathbf P^{(\ell_{i+1})}
\end{bmatrix}
=0,
\qquad
c_i(q)\in\mathbb R^{2N}.
\]
The family $\{c_i(q):q\in\mathcal F_i^\circ\}$ spans a subspace of
$\mathbb R^{2N}$ of dimension $r_i\leq2N$. Selecting a basis of this row
space and stacking its elements gives $C_i(\alpha_i)$ in
\eqref{eq:interface-finite-system}. Hence the finite system is equivalent
to imposing \eqref{eq:interface-gradient-proportionality} on
$\mathcal F_i^\circ$.
\hfill\hfill\qed
\end{pf}
\vspace{-1.5em}

Eq. \eqref{eq:interface-finite-system} couples the
control points, and hence the potentials, of neighboring cells. Under
the binary restriction $P_k\in\{0,-1\}$ of Lemma~\ref{lem:pk}, these
linear constraints need not admit a feasible solution. Retaining the
independent binary construction in each cell yields $X_{\ell_i}=X_{\ell_{i+1}}$, on the constant transition portion
$\mathcal F_i^\circ$, but does not in general ensure
$J_{X_{\ell_i}}=J_{X_{\ell_{i+1}}}$. To preserve continuity of the
control action, we therefore attenuate the feedforward term in a
neighborhood of each path interface. Therefore, for an active cell $\mathcal P_\ell$, let $d_\ell(q)$ denote the distance from $q$ to the nearest incident transition portion
$\mathcal F_i^\circ$ of the prescribed cell sequence. Hence
$d_\ell(q)=0$ on the interface and is positive in the cell interior.

\begin{proposition}
\label{prop:continuous-control-switching}
Consider a transition across $\mathcal F_i^\circ$ under the independent
binary construction of Lemma~\ref{lem:pk}, with
$\gamma_{P_{\ell_i}}=-1$, $\gamma_{P_{\ell_{i+1}}}=0$, and nonvanishing
gradients. Define
$\chi(d)=3(d/\delta)^2-2(d/\delta)^3$ for $0\leq d\leq\delta$ and
$\chi(d)=1$ for $d>\delta$, and consider
\begin{equation}
\label{eq:improved-control-action}
u_\ell(q,\dot q)
=
K\bigl(X_\ell(q)-\dot q\bigr)
+\chi(d_\ell(q))\dot X_\ell(q).
\end{equation}
Then \eqref{eq:improved-control-action} is continuous across
$\mathcal F_i^\circ$. In particular,
$X_{\ell_i}=X_{\ell_{i+1}}=n_i$ on $\mathcal F_i^\circ$, both one-sided
control limits are $K(n_i-\dot q)$, and \eqref{eq:control_action} is
recovered for $d_\ell\geq\delta$.
\end{proposition}
\vspace{-1.5em}

\begin{pf}
Since the potentials are constant on $\mathcal F_i^\circ$, their
gradients are normal to the facet. Under the binary construction of
Lemma~\ref{lem:pk} and the chosen orientation of $n_i$, both gradients
are negative multiples of $n_i$, and hence
$X_{\ell_i}=X_{\ell_{i+1}}=n_i$. At the interface,
$d_{\ell_i}=d_{\ell_{i+1}}=0$ and $\chi(0)=0$, so the feedforward term
vanishes from both sides and the control limits coincide
with $K(n_i-\dot q)$. For $d_\ell\geq\delta$, $\chi(d_\ell)=1$, and
\eqref{eq:improved-control-action} reduces to \eqref{eq:control_action}.
\hfill\hfill\qed
\end{pf}

\begin{corollary}
\label{cor:tracking-error-dynamics}
For $e:=X_\ell(q)-\dot q$, the closed-loop tracking error under
\eqref{eq:improved-control-action} satisfies
\begin{equation}
\label{eq:tracking-error-dynamics-new}
\dot e
=
-Ke+\bigl(1-\chi(d_\ell(q))\bigr)\dot X_\ell(q).
\end{equation}
Hence, outside the transition layers,
$e(t)=e^{-K(t-t_0)}e(t_0)$, so that $\dot q$ converges exponentially to
$X_\ell(q)$. Moreover, if
$\|\dot X_\ell(q)\|\leq\overline w$ inside a transition layer, then, for
any interval $[t_0,t]$ contained in that layer,
\begin{equation}
\label{eq:error-bound-transition}
\|e(t)\|
\leq
e^{-K(t-t_0)}\|e(t_0)\|
+
\frac{\overline w}{K}
\left(1-e^{-K(t-t_0)}\right).
\end{equation}
\end{corollary}
\vspace{-2.5em}

\begin{pf}
Differentiating $e=X_\ell-\dot q$ and substituting
\eqref{eq:improved-control-action} gives
\eqref{eq:tracking-error-dynamics-new}. For $\chi=1$, $\dot e=-Ke$.
Inside a transition layer,
$\dot e=-Ke+w$, where $w=(1-\chi)\dot X_\ell$ satisfies
$\|w\|\leq\overline w$; variation of constants yields
\eqref{eq:error-bound-transition}.
\hfill\hfill\qed
\end{pf}

\subsection{Forward invariance of non-exit facets}
\label{subsec:forward_invariance}

Smooth switching and velocity tracking do not guarantee
cell invariance for the double integrator, since the trajectory may reach
a non-exit facet with a finite outward velocity. Consider cell $\mathcal P_\ell$, with prescribed exit facet $\mathcal F_i$. For each non-exit facet $j\neq i$, choose
$\kappa_j>0$ and define
\begin{equation}
\label{eq:psi-barrier}
\psi_j(q,\dot q)
:=
\dot\beta_j(q)+\kappa_j\beta_j(q)
=
a_j^\top\dot q+\kappa_j\beta_j(q).
\end{equation}

\begin{proposition}
\label{prop:general-cell-invariance}
Let $t_0$ denote the entry time into $\mathcal P_\ell$, with
$q(t_0)\in\mathcal P_\ell$ and
$\psi_j(q(t_0),\dot q(t_0))\geq0$ for every $j\neq i$. If, for some
$\eta_j>0$,
\begin{equation}
\label{eq:barrier-condition}
\dot\psi_j+\eta_j\psi_j\geq0,
\qquad j\neq i,
\end{equation}
holds along the closed-loop trajectory while $\mathcal P_\ell$ is
active, then $\beta_j(q(t))\geq0$ for every $j\neq i$. Consequently,
the trajectory cannot leave $\mathcal P_\ell$ through a facet other than
$\mathcal F_i$.
\end{proposition}
\vspace{-1em}

\begin{pf}
From \eqref{eq:barrier-condition},
$\psi_j(t)\geq e^{-\eta_j(t-t_0)}\psi_j(t_0)\geq0$. Since
$\dot\beta_j+\kappa_j\beta_j=\psi_j\geq0$, it follows that
$\beta_j(q(t))\geq
e^{-\kappa_j(t-t_0)}\beta_j(q(t_0))\geq0$.
Hence every non-exit inequality is satisfied. Since
$\ddot\beta_j=a_j^\top u$, \eqref{eq:barrier-condition} is
equivalent to
\begin{equation}
\label{eq:second-order-invariance-condition}
a_j^\top u
+(\kappa_j+\eta_j)a_j^\top\dot q
+\kappa_j\eta_j\beta_j(q)
\geq0,
\qquad j\neq i.
\end{equation}
At a non-exit facet, $\beta_j=0$ and $\psi_j\geq0$ imply
$a_j^\top\dot q\geq0$, so the velocity is tangent to or interior-directed.
\hfill\hfill\qed
\end{pf}
\vspace{-1.5em}

\begin{corollary}
\label{cor:controller-cell-invariance}
Under the assumptions of Proposition~\ref{prop:general-cell-invariance},
the ideal and smoothed controllers satisfy the invariance condition if,
for every $j\neq i$,
\begin{multline}
\label{eq:invariance-controller}
\mkern-18mu K a_j^\top \mkern-3mu X_\ell \mkern-4mu
+\mkern-4mu \rho_\ell(q)a_j^\top\dot X_\ell
\mkern-5mu + \mkern-5mu (\kappa_j \mkern-5mu + \mkern-5mu \eta_j \mkern-5mu - \mkern-5mu K)a_j^\top \mkern-2mu \dot q \mkern-5mu
+\mkern-5mu \kappa_j\eta_j\beta_j(q)
\mkern-5mu\geq \mkern-5mu 0,
\end{multline}
where $\rho_\ell(q)=1$ for \eqref{eq:control_action} and
$\rho_\ell(q)=\chi(d_\ell(q))$ for
\eqref{eq:improved-control-action}.
\end{corollary}

\begin{pf}
The result follows by substituting the corresponding control law into
\eqref{eq:second-order-invariance-condition}.
\hfill\hfill\qed
\end{pf}

\begin{remark}
\label{rem:barrier-tuning}
Without loss of generality, the facet inequalities may be scaled such
that $\|a_j\|=1$, so that $\beta_j$ is the signed distance to the facet,
positive inside $\mathcal P_\ell$. The parameter $K$ sets the
velocity-tracking rate, $\kappa_j$ limits the admissible outward normal
velocity through $a_j^\top\dot q\geq-\kappa_j\beta_j$, and $\eta_j$
sets the convergence rate of $\psi_j$. On a constant non-exit facet,
$X_\ell=a_j$ and $a_j^\top\dot X_\ell=0$. Since $\beta_j=0$ and
$\psi_j\geq0$ imply $a_j^\top\dot q\geq0$,
$\kappa_j+\eta_j\geq K$ is a simple sufficient condition for
\eqref{eq:invariance-controller} at the facet, for both the ideal and
smoothed controllers. If \eqref{eq:second-order-invariance-condition}
is not satisfied by the nominal controller over the relevant state set,
it may instead be imposed as an affine constraint on $u$ in a
minimum-deviation quadratic program; see, e.g.,
\cite{rabiee2025softminimum}.
\eor
\end{remark}

\subsection{Cell sequence and closed-loop dynamics}
\label{subsec:cell_sequence}
\vspace{-1em}

Construct the feasible graph
$G=(\mathcal N,\mathcal E)$, whose nodes are the unique non-obstacle
cell facets,
\begin{equation}
\label{eq:graph-nodes}
\mathcal N=
\text{\texttt{unique}}
\left\{
\mathcal F_i:
\mathcal F_i\subset\partial\mathcal P_\ell,\ 
\mathcal F_i\not\subset\partial\mathcal O_j,\ \forall j
\right\},
\end{equation}
and whose edges connect distinct facets of the same cell,
\begin{equation}
\label{eq:graph-edges}
\mathcal E=
\left\{
(\mathcal F_i,\mathcal F_j)\in\mathcal N^2:
i\neq j,\ 
\exists\,\ell:\ 
\mathcal F_i,\mathcal F_j\subset\partial\mathcal P_\ell
\right\}.
\end{equation}
A shortest path gives the ordered facet sequence
\begin{equation}
\label{eq:graph-path}
\mkern-8mu(\mathcal F_{k_1},\ldots,\mathcal F_{k_M})
\mkern-2mu= \mkern-2mu
\text{\texttt{shortest\_path}}
(G,\mathcal F_{\mathrm{start}},\mathcal F_{\mathrm{end}}),
\end{equation}
with $\mathcal F_{k_1}=\mathcal F_{\mathrm{start}}$ and
$\mathcal F_{k_M}=\mathcal F_{\mathrm{end}}$. The corresponding ordered
cell sequence is
\begin{equation}
\label{eq:graph-cells}
\mathbb L=
(\ell_1,\ldots,\ell_{M-1}),
\quad
\mathcal F_{k_i},\mathcal F_{k_{i+1}}
\subset\partial\mathcal P_{\ell_i},
\end{equation}
where $i=1,\ldots,M-1$. The internal facets
$\mathcal F_{k_2},\ldots,\mathcal F_{k_{M-1}}$ are the switching
interfaces between consecutive cells. Hence, for an interior cell
$\mathcal P_{\ell_i}$, $i=2,\ldots,M-2$,
$
d_{\ell_i}(q)
=
\min\bigl\{
\operatorname{dist}(q,\mathcal F_{k_i}^\circ),
\operatorname{dist}(q,\mathcal F_{k_{i+1}}^\circ)
\bigr\}$. For the first and last cells,
$d_{\ell_1}(q)=\operatorname{dist}(q,\mathcal F_{k_2}^\circ)$ and
$d_{\ell_{M-1}}(q)=
\operatorname{dist}(q,\mathcal F_{k_{M-1}}^\circ)$, respectively.
Thus $d_{\ell_i}=0$ on either switching interface and is positive away
from them inside the cell. 
When \eqref{eq:invariance-controller} is satisfied by the smoothed
controller, the switched closed-loop dynamics are
\begin{equation}
\label{eq:switched_dynamics}
\ddot q
=
K\bigl(X_{\ell_i}(q)-\dot q\bigr)
+
\chi(d_{\ell_i}(q))\dot X_{\ell_i}(q),
\qquad
q\in\mathcal P_{\ell_i},
\end{equation}
where the active-cell index changes from $\ell_i$ to $\ell_{i+1}$ upon
crossing $\mathcal F_{k_{i+1}}^\circ$, $i=1,\ldots,M-2$. 

\begin{algorithm}[ht]
\caption{Trajectory computation}
\label{alg:path}
\algrenewcommand\algorithmicindent{0.6em}
\begin{algorithmic}[1]
\Require Cells $\mathcal P_\ell$, obstacles $\mathcal O_j$, start/end
facets $\mathcal F_{\mathrm{start}},\mathcal F_{\mathrm{end}}$,
B-spline parameters $(p,\lambda)$ satisfying
\eqref{eq:lambda_condition}, and controller parameters $(K,\delta)$.
\Ensure Trajectory $q^\star(t)$ along the selected cell sequence.
\State Construct $G$ and compute the facet sequence
\eqref{eq:graph-path} and cell sequence \eqref{eq:graph-cells}.

\For{$i=1:M-1$}
    \State Select $\mathcal P_{\ell_i}$ and exit facet
    $\mathcal F_{k_{i+1}}$.
    \State Compute $\beta_{\ell_i}$, $\varphi_{\ell_i}$, and the
    corresponding $(\theta_0,\theta_1)$.
    \State Obtain $\underline k,\overline k$ from
    \eqref{eq:k_bar_notation} and construct
    $\gamma_{\ell_i}$ and $\gamma_{P_{\ell_i}}$.
    \State Construct $X_{\ell_i}$ and the nominal $u_{\ell_i}$ from
    \eqref{eq:improved-control-action}.
\EndFor
\State Integrate \eqref{eq:dynamics} with the resulting switched
control to obtain $q^\star(t)$; when required,
\eqref{eq:second-order-invariance-condition} may be enforced through
the minimum-deviation QP of \cite{rabiee2025softminimum}. 
\end{algorithmic}
\end{algorithm}

\color{black}

\subsection*{Illustrative Example}
\label{sec:example}
\vspace{-1em}
Fig.~\ref{fig:path} illustrates a cluttered environment with $10$
obstacles
(\begin{tikzpicture}[baseline=-2pt, every node/.style={trapezium, draw}]
\node[trapezium,draw,trapezium left angle=75, trapezium right angle=45,
minimum width=2pt,line width=1pt,black,fill=red, fill opacity=0.3]
at(0.4,0.03){};
\end{tikzpicture}),
which induce a partition of the free space
(\begin{tikzpicture}[baseline=-2pt, every node/.style={trapezium, draw}]
\node[trapezium,draw,trapezium left angle=75, trapezium right angle=45,
minimum width=2pt,line width=1pt,black,fill=white] at(0.4,0.03){};
\end{tikzpicture})
into $41$ cells forming a polyhedral complex. The corresponding connectivity graph
contains $50$ nodes
(\begin{tikzpicture}[baseline=-2pt]
\node[circle,draw,minimum width=2pt,line width=1pt,black,fill=cyan,
fill opacity=0.2] at(0.4,0.03){};
\end{tikzpicture}),
representing feasible cell facets, and $77$ edges
(\begin{tikzpicture}[baseline=-2pt]
\draw[line width=1pt, blue, -Latex] (0,0.03) -- (0.5,0.03);
\end{tikzpicture}),
defined by \eqref{eq:graph-nodes}--\eqref{eq:graph-edges}; obstacle facets
and their incident edges are excluded. Hence, any graph path induces an
admissible sequence of free cells.

\begin{figure}[!ht]
\centering
\includegraphics[
    width=.85\columnwidth,
    trim=.5cm .5cm .5cm .5cm,
    clip
]{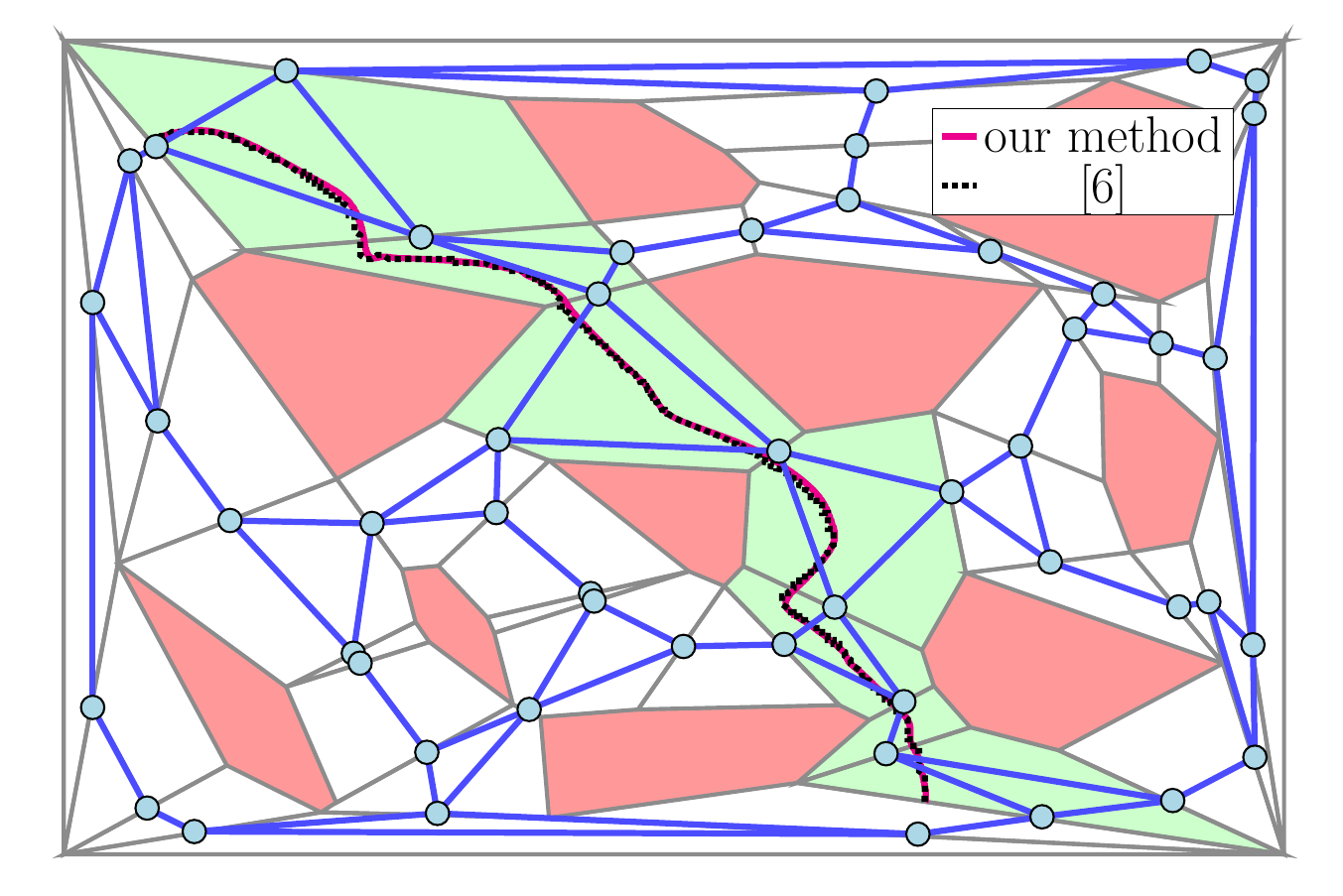}
\caption{Polyhedral decomposition of the cluttered environment and facet
connectivity graph \eqref{eq:graph-nodes}--\eqref{eq:graph-edges}. Graph
nodes are located at the corresponding facet centroids.}
\label{fig:path}
\end{figure}

Using facet-distance edge costs, Dijkstra's algorithm yields through
\eqref{eq:graph-path} the cell sequence
$\mathbb L=(1,3,9,10,15,16,17)$. The trajectories generated by
Alg.~\ref{alg:path}
(\begin{tikzpicture}[baseline=-2pt,scale=.6, draw=magenta, line width=1pt]
\draw (0.3,0.03) to[curve through = {(0.4,-0.14) (1,0.14)}] (1.2,0.03);
\end{tikzpicture}
-- proposed method,
\begin{tikzpicture}[baseline=-2pt,scale=.6, draw=black, densely dotted,
line width=1pt]
\draw (0.3,0.03) to[curve through = {(0.4,-0.14) (1,0.14)}] (1.2,0.03);
\end{tikzpicture}
-- \cite{conner2003construction})
remain inside the selected free cells
(\begin{tikzpicture}[baseline=-2pt, every node/.style={trapezium, draw}]
\node[trapezium,draw,trapezium left angle=75, trapezium right angle=45,
minimum width=2pt,line width=1pt,black,fill=green, fill opacity=0.2]
at(0.4,0.03){};
\end{tikzpicture})
and reach the target. For the representative trajectory considered below,
the terminal time is approximately $24.9\,\mathrm{s}$ in all four
configurations. We use $p=3$, $T=150$, hence
$\lambda=2\pi/T\simeq0.0419<0.0541$ as required by
\eqref{eq:lambda_condition}, and truncate the polylogarithm series after
$155$ terms.

Table~\ref{tab:path} illustrates the link between the unique cell facets (graph nodes) and the polyhedral cells that contain them: except for the initial (start) and final (target) facets, each intermediate facet serves simultaneously as an ``out'' facet for the preceding cell and an ``in'' facet for the current cell (the inclusion of a facet within a cell is characterized by the $0/1$ flag in the table).

\begin{table}[!ht]
  \centering
  \renewcommand{\arraystretch}{1.1}
  \begin{tabular}{cc*{7}{c}}
    \toprule
    & & \multicolumn{7}{c}{cell index ($\ell_i\in \mathbb L$, as in \eqref{eq:graph-cells})} \\
    \cmidrule(l){3-9}
    & & 1 & 3 & 9 & 10 & 15 & 16 & 17 \\
    \cmidrule(l){2-9}  
    \multirow{8}{*}{\rotatebox{90}{\shortstack{node index\\(cell facet, as in \eqref{eq:graph-path})}}} 
      &  3 & \cellcolor{blue!40}1 & 0 & 0 & 0 & 0 & 0 & 0 \\
      &  1 & \cellcolor{blue!40}1 & \cellcolor{blue!40}1 & 0 & 0 & 0 & 0 & 0 \\
      &  6 & 0 & \cellcolor{blue!40}1 & \cellcolor{blue!40}1 & 0 & 0 & 0 & 0 \\
      & 16 & 0 & 0 & \cellcolor{blue!40}1 & \cellcolor{blue!40}1 & 0 & 0 & 0 \\
      & 17 & 0 & 0 & 0 & \cellcolor{blue!40}1 & \cellcolor{blue!40}1 & 0 & 0 \\
      & 25 & 0 & 0 & 0 & 0 & \cellcolor{blue!40}1 & \cellcolor{blue!40}1 & 0 \\
      & 27 & 0 & 0 & 0 & 0 & 0 & \cellcolor{blue!40}1 & \cellcolor{blue!40}1 \\
      & 28 & 0 & 0 & 0 & 0 & 0 & 0 & \cellcolor{blue!40}1 \\
    \bottomrule
  \end{tabular}
  \caption{Example of path between nodes 3 and 28}
  \label{tab:path}
\end{table}

Fig.~\ref{fig:waterfall} illustrates the ``waterfall'' effect
for three cells
($\ell_3\mapsto\ell_4\mapsto\ell_5$): each harmonic surface guides the
agent through the prescribed exit facet, into the next cell.
\begin{figure}[!ht]
\centering
\includegraphics[
    width=.8\columnwidth,
    trim=.25cm .25cm .25cm .25cm,
    clip
]{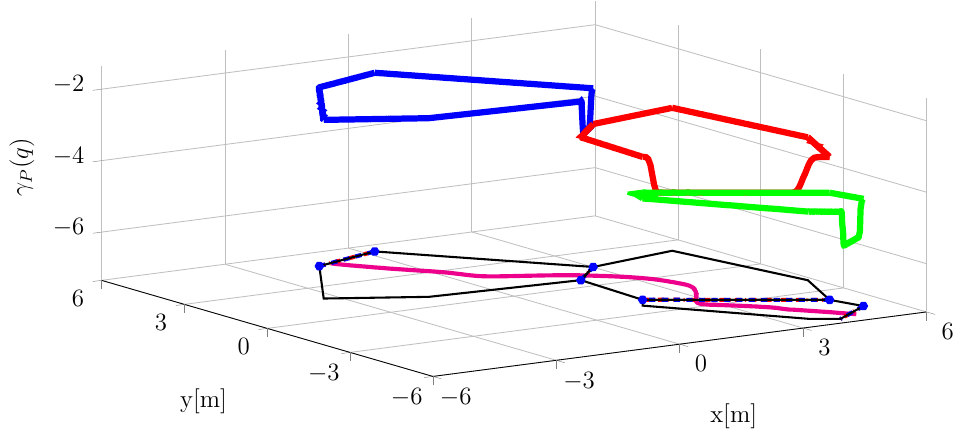}
\caption{``Waterfall'' effect along a path segment; surfaces and
trajectory are vertically shifted for clarity.}
\label{fig:waterfall}
\end{figure}

Table~\ref{tab:stats} compares trajectory-level control statistics after
uniform resampling in physical time.
\begin{table}[!ht]
\centering
\renewcommand{\arraystretch}{1}
\setlength{\tabcolsep}{3pt}
\small
\begin{tabular}{llcccc}
\toprule
\textbf{Boundary} & \textbf{Control}
& \multicolumn{2}{c}{$\|u\|_2$}
& \multicolumn{2}{c}{$\|u\|_\infty$}\\
\cmidrule(lr){3-4}\cmidrule(lr){5-6}
& & \textbf{Mean} & \textbf{Std}
  & \textbf{Mean} & \textbf{Std}\\
\midrule
B-spline & ideal
& 0.5274 & 1.0943 & 0.4723 & 0.9547\\
B-spline & smoothed
& 0.5222 & 1.0965 & 0.4678 & 0.9570\\
\cite{conner2003construction} & ideal
& 0.5238 & 1.0906 & 0.4674 & 0.9507\\
\cite{conner2003construction} & smoothed
& 0.5250 & 1.0951 & 0.4685 & 0.9552\\
\bottomrule
\end{tabular}
\caption{Trajectory-level control statistics for the 
trajectory in Fig.~\ref{fig:path}.}
\label{tab:stats}
\end{table}
\vspace{-1em}

The close values in Table~\ref{tab:stats} show that the proposed B-spline construction preserves the trajectory behavior of
\cite{conner2003construction} for both control laws. Its advantage is primarily structural: a finite set of control
points parametrizes a smooth boundary condition and its analytical
harmonic extension, while also allowing direct continuity and invariance
requirements. The effect is clearer locally. We sampled all non-exit facets along $\mathbb L$, including points close to their vertices. For $\dot q=X_\ell(q)$, the tracking term vanishes and $\|J_{X_\ell}X_\ell\|$ directly measures the spatial variation entering the feedforward action. With the ideal controller, the B-spline
construction reduces its RMS from $752.63$ to $715.95$ ($4.9\%$) and,
more notably, its maximum from $17254.60$ to $8462.55$ ($51\%$) relative
to \cite{conner2003construction}. As expected, \eqref{eq:improved-control-action} provides additional attenuation. For the B-spline construction, $P_{99}(\|u\|)$ decreases from $4024.78$ to $2005.23$ ($50.2\%$), while
$P_{99}(\|\rho_\ell J_{X_\ell}X_\ell\|)$ decreases from $3716.77$ to
$1746.29$ ($53.0\%$) and its RMS from $715.95$ to $455.84$ ($36.3\%$).

The values above are configuration-dependent and are therefore
illustrative rather than universal bounds. The B-spline formulation is
computationally more demanding than \cite{conner2003construction}, but
remains compatible with real-time execution: for a $0.25\,\mathrm{s}$
integration interval, the ideal and smoothed variants require about
$75.0$ and $79.7\,\mathrm{ms}$, respectively, corresponding to $0.300$
and $0.319\,\mathrm{s}$ of computation per simulated second. The controller is evaluated continuously within the ODE
right-hand side at the adaptive solver's internal points.
The code is available at
\url{gitlab.com/replan/replan-public}. It relies on CasADi
\cite{Andersson2019} for symbolic manipulation, automatic differentiation,
and numerical integration. 

\section{Conclusions}
\label{sec:con}
This paper has detailed the constructive underpinnings of a harmonic
function-based motion-planning framework in which the boundary conditions are induced by cardinal B-spline curves. The core contribution is a general treatment of control-point selection, enabling the systematic and flexible design of boundary behavior. This, in turn, yields smooth harmonic potentials on the unit disk that are pulled back to the cells of an arbitrary polyhedral decomposition of the configuration space. The resulting parametrization also permits inter-cell matching and forward-invariance conditions to be related directly to the B-spline control points, while a smoothed control law ensures continuous control action across independently constructed neighboring cells. Numerical results show that the proposed construction preserves the trajectory-level behavior of the reference harmonic formulation while reducing the largest local variations of the induced vector field. Future work will focus on further exploiting the flexibility of spline
descriptions, e.g., through variable knot vectors and alternative boundary constraints, on clarifying the relationship between cell decomposition and trajectory performance, and on investigating alternative pull maps $\varphi(q)$, including Schwarz--Christoffel conformal mappings, to reduce translational or angular distortion.


 \begin{IEEEbiography}[{\includegraphics[width=1in,height=1.2in,keepaspectratio]{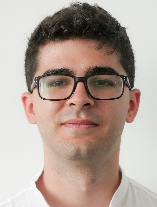}}]{Theodor-Gabriel Nicu}is a PhD Candidate and Teaching Assistant in the Automatic Control and Systems Engineering Department of the National University of Science and Technology Politehnica Bucharest, Romania. He received BSc and MSc degrees in Systems Engineering from Politehnica University of Bucharest, in 2020 and 2022, respectively. Since 2024 he is also affiliated with the CAMPUS Research Institute. His current research interests are focused on the Artificial Potential Field approach for obstacle collision avoidance algorithms, B-Spline and also mixed-integer representations.
 \end{IEEEbiography}
\vspace{-1.5em}
 \begin{IEEEbiography}[{\includegraphics[width=1in,height=1.2in,keepaspectratio]{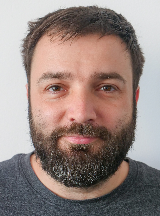}}]{Florin Stoican}is Full Professor in the Automatic Control and Systems Engineering Department of the National University of Science and Technology Politehnica Bucharest, Romania. He received his Ph.D. at Supélec in 2011 and obtained the habilitation thesis in set-theoretic control in 2018. Between 2014-2015 he pursued a post-doctoral fellowship at NTNU, Norway. Since 2024 he is also affiliated with the CAMPUS Research Institute. His expertise covers MPC, robust and set-theoretic control, mixed-integer optimization and factor graph optimization.
 \end{IEEEbiography}
\vspace{-1.5em}
 \begin{IEEEbiography}[{\includegraphics[width=1in,height=1.25in,keepaspectratio]{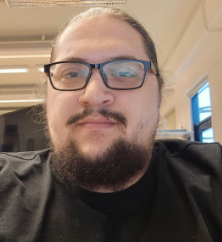}}]{Daniel-Mihail Ioan}is Assistant Professor in the Automatic Control and Systems Engineering Department of the National University of Science and Technology Politehnica Bucharest, Romania. He received his Ph.D. at Univ. Paris-Saclay and is also a former postdoctoral researcher at ENSTA Bretagne. Since 2024 he is also affiliated with the CAMPUS Research Institute. His area of expertise lies within the optimisation-based control framework and also motion planning and robotic implementation.
 \end{IEEEbiography}
\vspace{-1.5em}
 \begin{IEEEbiography}[{\includegraphics[width=1in,height=1.25in,keepaspectratio]{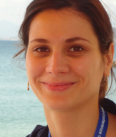}}]{Ionela Prodan}received the B.S. degree in Automation and Applied Informatics from Univ. Politehnica of Bucharest, Romania in 2009, and her Ph.D. in Control Engineering from Supélec, Gif-sur-Yvette, France in 2012. She pursued a post-doctoral fellowship at EDF, École Centrale Paris, France. Since 2014, she is also Associate Professor at the LCIS Laboratory of Grenoble INP, Univ. Grenoble Alpes, Valence, France. In 2020, she obtained her habilitation thesis. Dr. Prodan’s research interests focus on constrained optimization-based control, mixed-integer programming, differential flatness,
set-theoretic methods, and their application to motion planning for autonomous
vehicles, energy management in building-scale DC microgrids and multi-physics
systems. She is an Associate Editor for the EUCA European Control Conference since 2022, an Associate Editor for the IEEE Transactions on Control Systems Technology journal since 2023, Control Engineering Practice journal since 2024.
 \end{IEEEbiography}

\end{document}